\documentclass[12pt]{article}

\usepackage{amsmath,amssymb,amsthm,epsfig,mathrsfs}
\usepackage{hyperref}
\usepackage{authblk}
\usepackage{fullpage}

\newcommand{\vpc}{b}

\newcommand{\Aut}{\mathop{\rm Aut}\nolimits}
\newcommand{\Isom}{\mathop{\rm Isom}\nolimits}
\newcommand{\End}{\mathop{\rm End}\nolimits}
\newcommand{\mmod}{\,\mathrm{mod}\,}

\newcommand{\Tr}{\mathop{\rm Tr}}

\newcommand{\A}{\mathcal A}
\newcommand{\B}{\mathcal B}

\newcommand{\Z}{\mathbb Z}
\newcommand{\N}{\mathbb N}
\newcommand{\Q}{\mathbb Q}
\newcommand{\R}{\mathbb R}
\newcommand{\Sym}{\mathop{\rm Sym}\nolimits}
\newcommand{\LL}{\mathcal L}

\newtheorem{theorem}{Theorem}[section]
\newtheorem{algorithm}{Algorithm}
\newtheorem{corollary}[theorem]{Corollary}
\newtheorem{proposition}[theorem]{Proposition}
\newtheorem{lemma}[theorem]{Lemma}

\theoremstyle{definition}
\newtheorem{definition}{Definition}

\title{Solenoid Maps, Automatic Sequences, Van Der Put Series, and Mealy-Moore Automata}

\author[1]{Rostislav Grigorchuk}
\affil[1]{               Department of Mathematics\\
               Texas A\&M University\\
               College Station, TX 77843-3368\\
               \href{mailto:grigorch@math.tamu.edu}{grigorch@math.tamu.edu}}

\author[2]{Dmytro Savchuk}
\affil[2]{Department of Mathematics and Statistics\\
               University of South Florida\\
               4202 E Fowler Ave\\
               Tampa, FL 33620-5700\\
               \href{mailto:savchuk@usf.edu}{savchuk@usf.edu}}

\begin{document}

\maketitle

\begin{abstract}
The ring $\mathbb Z_d$ of $d$-adic integers has a natural interpretation as the boundary of a rooted $d$-ary tree $T_d$. Endomorphisms of this tree (i.e. solenoid maps) are in one-to-one correspondence with 1-Lipschitz mappings from $\mathbb Z_d$ to itself and automorphisms of $T_d$ constitute the group $\mathrm{Isom}(\Z_d)$. In the case when $d=p$ is prime, Anashin showed in~\cite{anashin:automata12} that $f\in\mathrm{Lip}^1(\Z_p)$ is defined by a finite Mealy automaton if and only if the reduced coefficients of its van der Put series constitute a $p$-automatic sequence over a finite subset of $\Z_p\cap\mathbb Q$. We generalize this result to arbitrary integer $d\geq 2$, describe the explicit connection between the Moore automaton producing such sequence and the Mealy automaton inducing the corresponding endomorphism. Along the process we produce two algorithms allowing to convert the Mealy automaton of an endomorphism to the corresponding Moore automaton generating the sequence of the reduced van der Put coefficients of the induced map on $\Z_d$ and vice versa. We demonstrate examples of applications of these algorithms for the case when the sequence of coefficients is Thue-Morse sequence, and also for one of the generators of the standard automaton representation of the lamplighter group.
\end{abstract}

\section{Introduction}

Continuous self-maps of the ring $\Z_p$ of $p$-adic integers is the object of study of $p$-adic analysis and $p$-adic dynamics.
Among all continuous functions $\Z_p\to\Z_p$ there is an natural subclass of $1$-Lipschitz functions that do not increase distances between points of $\Z_p$. These functions appear in many contexts and have various names in the literature. For example, Bernstein and Lagarias in the paper devoted to the Collatz ``$3n+1$'' conjecture call them solenoidal maps~\cite{bernstein_l:3n_plus_1_96}, Anashin in~\cite{ana:erg} (see also~\cite{anashin_ky:characterization_of_ergodic11}) studied the conditions under which these functions act ergodically on $\Z_p$. For us such functions are especially important because they act on regular rooted trees by endomorphisms (or automorphisms in the invertible case). Topologically, $\Z_p$ is homeomorphic to the Cantor set which, in turn, can be identified with the boundary $X^\infty$ of a rooted $p$-ary tree $X^*$, whose vertices are finite words over the alphabet $X=\{0,1,\ldots,p-1\}$. Namely, we identify a $p$-adic number $x_0+x_1p+x_2p^2+\cdots$ with the point $x_0x_1x_2\ldots\in X^\infty$. (For the language of rooted trees and group actions on them see~\cite{gns00:automata,grigorchuk-s:standrews}).

Under this identification Nekrashevych, Sushchansky, and the first author~\cite[Proposition~3.7]{gns00:automata} showed that a continuous map from $\Z_p$ to itself induces a (graph) endomorphism of the tree $X^*$ precisely when it is 1-Lipschitz. Furthermore,
it is an easy but not so well-known observation that the group $\Isom(\Z_p)$ of isometries of $\Z_p$ is naturally isomorphic to the group $\Aut(X^*)$ of automorphisms of a rooted $p$-ary tree. As such, the groups $\Isom(\Z_p)$ contain many exotic groups that provided counterexamples to several long standing conjectures and problems in group theory~\cite{grigorch:burnside,grigorch:milnor,grigorch_lsz:atiyah,grigorch_z:basilica} and have connections to other areas of mathematics, such as holomorphic dynamics~\cite{nekrash:self-similar,bartholdi_n:rabbit}, combinatorics~\cite{grigorchuk-s:hanoi-cr}, analysis on graphs~\cite{grigorch_s:hanoi_spectrum}, computer science~\cite{cain:automaton_semigroups,miasnikov_s:cayley_automatic11,miasnikov_s:automatic_graph}, cryptography~\cite{myasnikov_u:random_subgroups08,myasnikov_su:non_commutative_crypto_book11,garzon_z:crypto,petrides:cryptoanalysis_grigorchuk} and coding theory~\cite{cull_n:hanoi_codes99,grigorchuk-s:hanoi-cr}. In a similar way one can characterize the group $\Isom(\Q_p)$ of isometries
of the field $\Q_p$ of $p$-adic numbers as the group of automorphisms of a regular (not rooted) $(p+1)$-ary tree that fix pointwise one selected end of this tree.

To describe important subgroups of $\Isom(\Z_p)$ and establish their properties, the languages of self-similar groups and semigroups initiated in~\cite{grigorch:burnside} and developed in the last four decades (see survey papers~\cite{gns00:automata,bartholdi_gn:fractal} and the book~\cite{nekrash:self-similar}), and Mealy automata have proved to be very effective. On the other hand, these tools were not widely used by researchers studying $p$-adic analysis and $p$-adic dynamics. There are only few papers that build bridges between the two worlds. The first realization of an affine transformation of $\Z_p$ by a finite Mealy automaton was constructed by Bartholdi and \v Suni\'c in~\cite{bartholdi_s:bsolitar}. Ahmed and the second author in~\cite{ahmed_s:polynomial_ergodicity} described automata defining polynomial functions $x\mapsto f(x)$ on $\Z_d$, where $f\in\Z[x]$, and using the language of groups acting on rooted trees deduced conditions for ergodicity of the action of $f$ on $\Z_2$ obtained by completely different methods by Larin~\cite{larin:tr}.
In~\cite{anashin:automata12} Anashin proved an excellent result relating finiteness of the Mealy automaton generating an endomorphism of the $p$-ary tree with automaticity of the sequence of reduced van der Put coefficients of the induced functions on $\Z_p$, which will be discussed below in details. Automatic sequences represent an important area at the conjunction of computer science and mathematics. Some of the famous examples of automatic sequences include Thue-Morse sequence and Rudin-Shapiro sequence defining space filling curves. We refer the reader to~\cite{allouche:automatic_sequences03} for details. Recent applications of automatic sequences in group theory include~\cite{grigorch_lns:self-similar_groups_automatic_sequences16,grigorch_ln:spectra_schreier_schroedinger18}.

As in the real analysis, one of the effective ways to study functions $\Z_p\to\Z_p$ is to decompose them into series with respect to some natural basis in the space of continuous functions $C(\Z_p)$ from $\Z_p$ to itself. Two of the most widely used bases of this space are the Mahler basis and the van der Put bases~\cite{mahler:p-adic,schikhof:ultrametric84}. In the more general settings of the spaces of continuous functions from $\Z_p$ to a field, several other bases have been used in the literature: Walsh basis~\cite{walsh:basis23}, Haar basis (used in group theory context, for example, in~\cite{bartholdi_g:spectrum}), Kaloujnine basis~\cite{grigorch_lns:self-similar_groups_automatic_sequences16}. In this paper we will deal with the van der Put basis, which is made of functions $\chi_n(x)$, $n\geq 0$ that are characteristic functions of cylindrical subsets of $\Z_p$ consisting of all elements that have the $p$-adic expansion of $n$ as a prefix. Each continuous function $f\in C(\Z_p)$ can be decomposed uniquely as
\begin{equation}
f(x)=\sum_{n\geq 0}B^f_n\chi_n(x),
\end{equation}
where the coefficients $B^f_n$ are elements of $\Z_p$ which we call van der Put coefficients. A function $f\colon\Z_p\to\Z_p$ is 1-Lipschitz if and only if its van der Put coefficients can be represented as $B^f_n=\vpc^f_nd^{\lfloor\log_dn\rfloor}$ for all $n>0$, where $\vpc^f_n\in\Z_p$~\cite{anashin_ky:characterization_of_ergodic11}. We call $\vpc^f_n$ the \emph{reduced van der Put coefficients} (see Section~\ref{sec:continuous} for details).

The main results of the present paper are the following two theorems, in which $d\geq 2$ is an arbitrary (not necessarily prime) integer.

\begin{theorem}
\label{thm:characterization}
Let $g\in\End(X^*)$ be an endomorphism of the rooted tree $X^*$, where $X=\{0,1,\ldots,d-1\}$. Then $g$ is finite state if and only if the following two conditions hold for the transformation $\hat g$ of $\Z_d$ induced by $g$:
\begin{enumerate}
\item[(a)] the sequence $(\vpc_n^{\hat g})_{n\geq1}$ of reduced van der Put coefficients of $\hat g$ consists of finitely many eventually periodic elements from $\Z_d$;
\item[(b)] $(\vpc_n^{\hat g})_{n\geq1}$ is $d$-automatic.
\end{enumerate}
\end{theorem}

For the case of prime $d=p$, Theorem~\ref{thm:characterization} was proved by Anashin in~\cite{anashin:automata12} using a completely different from our approach.  The proof from~\cite{anashin:automata12} does not provide a direct connection between Mealy automaton of an endomorphism of $X^*$ and the Moore automaton of the corresponding sequence of its reduced van der Put coefficients. Our considerations are based on understanding the connection between the reduced van der Put coefficients of an endomorphism and of its sections at vertices of the rooted  tree via the geometric notion of a portrait. This connection, summarized in the next theorem,  bears a distinct geometric flavor  and gives a way to effectively relate the corresponding Mealy and Moore automata.

\begin{theorem}
\label{thm:main}
Let $X=\{0,1,\ldots,d-1\}$ be a finite alphabet identified with $\Z/d\Z$.\\[-5mm]
\begin{itemize}
\item[(a)] Given an endomorphism $g$ of the tree $X^*$, defined by the finite Mealy automaton, there is an explicit algorithmic procedure given by Theorem~\ref{thm:covering} and Algorithm~\ref{alg:mealy_to_moore}, that constructs the finite Moore automaton generating the sequence $(\vpc_n^g)_{n\geq 0}$ of reduced van der Put coefficients of $g$.
\item[(b)] Conversely, given a finite Moore automaton generating the sequence $(c_n)_{n\geq 0}$ of eventually periodic $d$-adic integers, there is an explicit algorithmic procedure given by Theorem~\ref{thm:covering2} and Algorithm~\ref{alg:moore_to_mealy}, that constructs the finite Mealy automaton of an endomorphism $g$ with the reduced van der Put coefficients satisfying $\vpc_n^g=c_n$ for all $n\geq 0$.
\item[(c)] Both constructions are dual to each other in a sense that the automata produced by them cover the input automata as labelled graphs (see Section~\ref{sec:relation} for the exact definition).
\end{itemize}
\end{theorem}

Theorem~\ref{thm:main} opens up a new approach to study automatic sequences by means of (semi)groups acting on rooted trees, and vice versa, to study endomorphisms of rooted trees via the language of Moore automata.

In~\cite{anashin:automata12} Anashin, using  the  famous Christol's  characterization of  the $p$-automatic  sequences in terms of algebraicity of the corresponding power series, suggested  another  version of  the  main  result of  his  paper  (i.e., of Theorem~\ref{thm:characterization}  in the  case  of  prime $d$).  The  authors  are not  aware  of  the  existence of  the  analogue  of  Christol's  theorem  in  the  situation of  $d$-automaticity  when  $d$  is  not prime. The  first  question  arises  in  what  sense to  mean the  algebraicity of  function  when  Field  $\Q_p$ is replaced  by  the ring  $\Q_d$ of $d$-adic numbers. The  authors  do not  exclude  that the  extension of  Christol's  theorem is possible and leave this question for the future.

The paper is organized as follows. Section~\ref{sec:pre} introduces necessary notions related to Mealy automata and actions on rooted trees. Section~\ref{sec:continuous} recalls how to represent a continuous function $\Z_d\to\Z_d$ by a van der Put series. We consider automatic sequences and define their portraits and sections in Section~\ref{sec:automatic}. The crucial argument relating van der Put coefficients of endomorphisms and their sections is given in~\ref{sec:portraits}. Section~\ref{sec:proof} contains the proof of Theorem~\ref{thm:characterization}. The algorithms relating Mealy and Moore automata associated with an endomorphism of $X^*$ and constituting the proof of Theorem~\ref{thm:main}, are given in Section~\ref{sec:relation}. Finally, two examples are worked out in full details in Section~\ref{sec:examples} that concludes the paper.

\section{Mealy Automata and endomorphisms of rooted trees}

\label{sec:pre}
We start this section by introducing the notions and terminology of endomorphisms and automorphisms of regular rooted trees and transformations generated by Mealy automata. For more details, the reader is referred to~\cite{gns00:automata}.

Let $X=\{0,1,\ldots ,d-1\}$ be a finite alphabet with $d\geq 2$ elements (called letters) and let $X^*$ denote the set of all finite words over $X$. The set $X^*$ can be equipped with the structure of a rooted $d$-ary tree by declaring that $v$ is adjacent to $vx$ for every $v\in X^*$ and $x\in X$. Thus finite words over $X$ serve as vertices of the tree. The empty word corresponds to the root of the tree and for each positive integer $n$ the set $X^n$ corresponds to the $n$-th level of the tree. Also the set $X^\infty$ of infinite words over $X$ can be identified with the \emph{boundary} of the tree $X^*$ consisting of all infinite paths in the tree without backtracking initiating at the root. We will consider endomorphisms and automorphisms of the tree $X^*$ (i.e., the maps and bijections of $X^*$ that preserve the root and the adjacency of vertices). We will sometimes denote the tree $X^*$ sa $T_d$. The semigroup of all endomorphisms of $T_d$ is denoted by $\End(T_d)$ and the group of all automorphisms of $T_d$ is denoted by $\Aut(T_d)$. To operate with such objects, we will use the language of Mealy automata.

\begin{definition}
A \emph{Mealy automaton} (or simply \emph{automaton}) is a 4-tuple \[(Q,X,\delta,\lambda),\]
where
\begin{itemize}
\item $Q$ is a set of states
\item $X$ is a finite alphabet (not necessarily $\{0,1,\ldots,d-1\}$)
\item $\delta\colon Q\times X\to Q$ is the \emph{transition function}
\item $\lambda\colon Q\times X\to X$ is the \emph{output function}.
\end{itemize}
If the set of states $Q$ is finite, the automaton is called \emph{finite}. If for every state $q\in Q$ the output function $\lambda_q(x)=\lambda(q,x)$ induces
a permutation of $X$, the automaton $\A$ is called \emph{invertible}. Selecting a state $q\in Q$ produces an \emph{initial automaton} $\A_q$, which formally is a $5$-tuple $(Q,X,\delta,\lambda,q)$.
\end{definition}

Here we consider automata with the same input and output alphabets.

Automata are often represented by their \emph{Moore diagrams}. The Moore diagram of automaton $\A=(Q,X,\delta,\lambda)$ is a directed graph in which the vertices are in bijection with the states of $Q$ and the edges have the form $q\stackrel{x|\lambda(q,x)}{\longrightarrow}\delta(q,x)$ for $q\in Q$ and $x\in X$. Figure~\ref{fig:lamp_aut} shows the Moore diagram of the automaton $\A$ that, as will be explained later, generates the lamplighter group $\mathcal L=(\Z/2\Z)\wr\Z$.

\begin{figure}[h]
\begin{center}
\epsfig{file=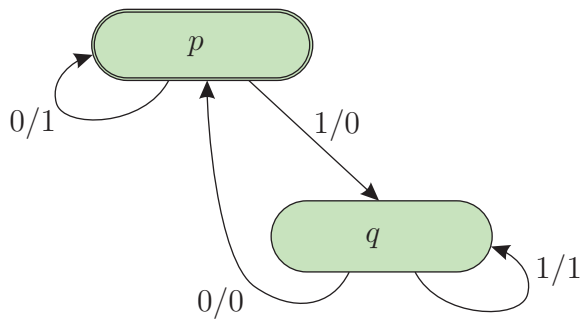}
\caption{Mealy automaton generating the lamplighter group $\mathcal L$\label{fig:lamp_aut}}
\end{center}
\end{figure}

Every initial automaton $\A_q$ induces an endomorphism of $X^*$, which will be also denoted by $\A_q$, defined as follows. Given a word
$v=x_1x_2x_3\ldots x_n\in X^*$, it scans its first letter $x_1$ and outputs $\lambda(q,x_1)$. The rest of the word is handled similarly by the initial automaton $\A_{\delta(q,x_1)}$. So we can actually extend the functions $\delta$ and $\lambda$ to $\delta\colon Q\times X^*\to Q$ and $\lambda\colon  Q\times X^*\to X^*$ via the equations
\[\begin{array}{l}
\delta(q,x_1x_2\ldots x_n)=\delta(\delta(q,x_1),x_2x_3\ldots x_n),\\
\lambda(q,x_1x_2\ldots x_n)=\lambda(q,x_1)\lambda(\delta(q,x_1),x_2x_3\ldots x_n).\\
\end{array}
\]

The boundary $X^\infty$ of the tree is endowed with a natural topology in which two infinite words are close if they have large common prefix. With this topology $X^\infty$ is homeomorphic to the Cantor set. Each endomorphism (automorphism) of $X^*$ naturally induces a continuous transformation (homeomorphism) of $X^\infty$.

\begin{definition}
The semigroup (group) generated by all states of an automaton $\A$ viewed as endomorphisms (automorphisms) of the rooted tree $X^*$ under the operation of composition is called an \emph{automaton semigroup (group)} and is denoted by $\mathbb{S}(\A)$ (respectively $\mathbb{G}(\A)$).
\end{definition}

In the definition of the automaton, we do not require the set $Q$ of states to be finite. With this convention, the notion of an automaton group is equivalent to the notions of \emph{self-similar group}~\cite{nekrash:self-similar} and \emph{state-closed group}~\cite{nekrash_s:12endomorph}. However, most of the interesting examples of automaton (semi)groups are finitely generated (semi)groups defined by finite automata.

Let $g\in\End(X^*)$ and $x\in X$. For any $v\in X^*$ we can write \[g(xv)=g(x)v'\] for some $v'\in X^*$. Then the map $g|_x\colon X^*\to X^*$ given by \[g|_x(v)=v'\] defines an endomorphism of $X^*$ which we call the \emph{state} (or \emph{section}) of $g$ at vertex $x$. We can inductively extend the definition of section at a letter $x\in X$ to section at any vertex $x_1x_2\ldots x_n\in X^*$ as follows. \[g|_{x_1x_2\ldots x_n}=g|_{x_1}|_{x_2}\ldots|_{x_n}.\]

We will adopt the following convention throughout the paper. If $g$ and $h$ are elements of some (semi)group acting on a set $Y$ and $y\in Y$, then $$gh(y)=h(g(y)).$$
Hence the state $g|_v$ at $v\in X^*$ of any product $g=g_1g_2\cdots g_n$, where $g_i\in\Aut(X^*)$ for $1\leq i\leq n$, can be computed as follows:
$$g|_v=g_1|_v g_2|_{g_1(v)}\cdots
g_n|_{g_1g_2\cdots g_{n-1}(v)}.$$

Also we will use the language of the wreath recursions. For each automaton semigroup $G$ there is a natural embedding
\[G\hookrightarrow G \wr \Tr(X),\]
where $\Tr(X)$ denotes the semigroup of all selfmaps of set $X$. This embedding is given by
\begin{equation}
\label{eq:wreath}
G\ni g\mapsto (g_0,g_1,\ldots,g_{d-1})\sigma_g\in G\wr \Tr(X),
\end{equation}
where $g_0,g_1,\ldots,g_{d-1}$ are the states of $g$ at the vertices of the first level, and $\sigma_g$ is the transformation of $X$ induced by the action of $g$ on the first level of the tree. If $\sigma_g$ is the trivial transformation, it is customary to omit it in~\eqref{eq:wreath}. We call $(g_0,g_1,\ldots,g_{d-1})\sigma_g$ the \emph{decomposition of $g$ at the first level} (or the \emph{wreath recursion of $g$}).

In the case of the automaton group $G=\mathbb{G}(\A)$, the embedding~\eqref{eq:wreath} is actually the embedding into the group $G\wr\Sym(X)$.

The decomposition at the first level of all generators $\A_q$ of an automaton semigroup $\mathbb{S}(\A)$ under the embedding~\eqref{eq:wreath} is called the \emph{wreath recursion} defining the semigroup. It is a convenient language when doing computations involving the states of endomorphisms. Indeed, the products endomorphisms and inverses of automorphisms can be found as follows. If $g\mapsto (g_0,g_1,\ldots,g_{d-1})\sigma_g$ and $h\mapsto (h_0,h_1,\ldots,h_{d-1})\sigma_h$ are two elements of $\End(X^*)$, then $$gh=(g_0h_{\sigma_g(0)},g_1h_{\sigma_g(1)},\ldots,g_{d-1}h_{\sigma_g(d-1)})\sigma_g\sigma_h$$ and in the case when $g$ is an automorphism, the wreath recursion of $g^{-1}$ is
$$g^{-1}=(g^{-1}_{\sigma_g^{-1}(0)},g^{-1}_{\sigma_g^{-1}(1)},\ldots,g^{-1}_{\sigma_g^{-1}(d-1)})\sigma_g^{-1}.$$

\section{Continuous maps from $\Z_d$ to $\Z_d$}
\label{sec:continuous}
In this section we will recall how to represent every continuous function $f\colon\Z_d\to\Z_d$ by its van der Put series. For details when $d=p$ is prime we refer the reader to Schikhof's book~\cite{schikhof:ultrametric84} and for needed facts about the ring of $d$-adic integers we recommend~\cite[Section 4.2]{goresky_k:alg_shift_register_sequences12} and~\cite{katok:p-adic_analysis07}. Here we will relate the coefficient of these series to the vertices of the rooted $d$-ary tree, whose boundary is identified with $\Z_d$.

First, we recall that the ring of $d$-adic integers $\Z_d$ for arbitrary (not necessarily prime) $d$ is defined as the set of all formal sums
\[\Z_d=\bigl\{a_0+a_1d+a_2d^2+\cdots\colon a_i\in\{0,1,\ldots,d-1\}=\Z/d\Z, i\geq 0\bigr\},\]
where addition and multiplication are defined in the same way as in $\Z_p$ for prime $p$ taking into account the carry over. Also, the ring $\Q_d$ of $d$-adic numbers can be defined as the full ring of fractions of $\Z_d$, but we will only need to use elements of $\Z_d$ below. Algebraically, if $d=p_1^{n_1}p_2^{n_2}\cdots p_k^{n_k}$ is the decomposition of $d$ into the product of primes, then
\[\Z_d=\Z_{p_1}\times\Z_{p_2}\times\cdots\times\Z_{p_k}\qquad \text{ and }\qquad \Q_d=\Q_{p_1}\times\Q_{p_2}\times\cdots\times\Q_{p_k}.\]

As stated in the introduction, for the alphabet $X=\{0,1,\ldots,d-1\}$ we identify $\Z_d$ with the boundary $X^\infty$ of the rooted $d$-ary regular tree $X^*$ in a natural way, viewing a $d$-adic number $x_0+x_1d+{x_2d^2+\cdots}$ as a point $x_0x_1x_2\ldots\in X^\infty$. This identification gives rise to an embedding of $\N_0=\N\cup\{0\}$ into $X^*$ via $n\mapsto [n]_d$, where $[n]_d$ denotes the word over $X$ representing the expansion of $n$ in base $d$ written backwards (so that, for example, $[6]_2=011$). There are two standard ways to define the image $[0]_d$ of $0\in\N_0$: one can define it to be either the empty word $\varepsilon$ over $X$ of length 0, or a word $0$ of length 1. These two choices will give rise later to two similar versions of the van der Put bases in the space of continuous functions from $\Z_d$ to $\Z_d$, that we will call Mahler and Schikhof versions. Throughout the paper we will use Mahler's version and, unless otherwise stated, we will define $[0]_d=0$ (the word of length 1). However, we will state some of the results for Schikhof's version as well. Note, that the image of $\N\cup\{0\}$ consists of all vertices of $X^*$ that do not end with $0$, and the vertex $0$ itself. We will called these vertices \emph{labelled}. For example, the labelling of the binary tree is shown in Figure~\ref{fig:binary_labelling}. The inverse of this embedding, with a slight abuse of notation as the notation does not explicitly mention $d$, we will denote by bar $\overline{\phantom a}$. In other words, if $u=u_0u_1\ldots u_n\in X^*$, then $\overline u=u_0+u_1d+\cdots+u_nd^{n}\in\N_0$. We note that the operation $u\mapsto\overline u$ is not injective as $\overline u=\overline{u0^k}$ for all $k\geq 0$.

\begin{figure}[h]
\begin{center}
\epsfig{file=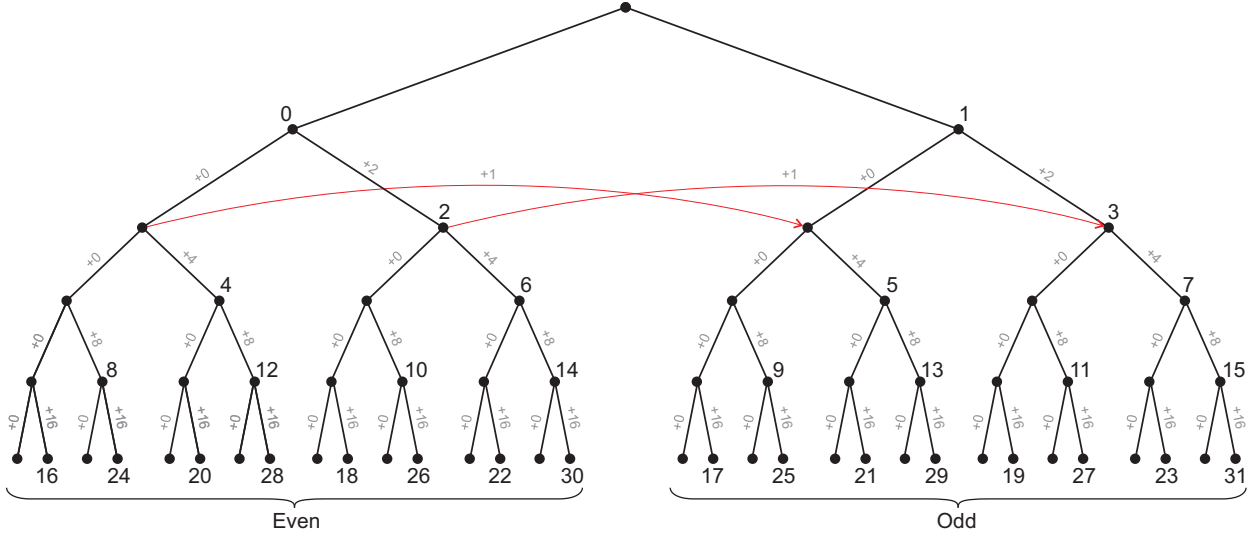, width=\textwidth}
\caption{Labelling of vertices of a binary tree by elements of $\N_0$\label{fig:binary_labelling}}
\end{center}
\end{figure}

Under this notation we can also define for each $n\geq 0$ a cylindrical subset $[n]_dX^\infty\subset\Z_d$ that consists of all $d$-adic integers that have $[n]_d$ as a prefix. Geometrically thus set can be envisioned as the boundary of the subtree of $X^*$ hanging down from the vertex $[n]_d$.

For $n>0$ with the $d$-ary expansion $n=x_0+x_1d+\cdots +x_kd^k$, $x_k\neq 0$, we define $n\_=n-x_kd^k$. Geometrically, $n\_$ is the label of the closest to $n$ labelled vertex in $X^*$ along the unique path from $n$ to the root of the tree. For example, for $n=22$ we have $[n]_2=01101$, so $[n\_]_2=011$ and $n\_=6$.

We are ready to define the decomposition of a continuous function $f\colon\Z_d\to\Z_d$ into a van der Put series. For each such function there is a unique sequence $(B^f_n)_{n\geq 0}$, $B^f_n\in\Z_d$ of $d$-adic integers such that for each $x\in\Z_d$ the following expansion
\begin{equation}
\label{eq:mahler}
f(x)=\sum_{n\geq 0}B^f_n\chi_n(x),
\end{equation}
holds, where $\chi_n(x)$ is the characteristic function of the cylindrical set $[n]_dX^\infty$ with values in $\Z_d$. The coefficients $B^f_n$ are called the \emph{van der Put coefficients} of $f$ and are computed as follows:

\begin{equation}
\label{eqn:B_n}
B^f_n=\left\{
\begin{array}{ll}
f(n),&\text{if}\ \ 0\leq n<d,\\
f(n)-f(n\_),&\text{if}\ \ n\geq d.
\end{array}
\right.
\end{equation}

This is the decomposition with respect to the orthonormal van der Put basis $\{\chi_n(x)\colon n\geq 0\}$ of the space $C(\Z_d)$ of continuous functions from $\Z_d$ (as $\Z_d$-module) to itself, as given in Mahler's book~\cite{mahler:p-adic}, and also used in~\cite{anashin_ky:characterization_of_ergodic11}.
In the literature this basis is considered only when $d=p$ is a prime number, and is, in fact, an orthonormal basis of a larger space $C(\Z_p\to K)$ of continuous functions from $\Z_p$ to a normed field $K$ containing the field of $p$-adic rationals $\mathbb Q_p$. However, the given decomposition works in our context with all the proofs identical to the ``field'' case.

To avoid possible confusion we note that there is another standard version of the van der Put basis $\{\tilde\chi_n(x)\colon n\geq 0\}$ used, for example, in Schikhof's book~\cite{schikhof:ultrametric84}. We will call this version of a basis Schikhof's version. In this basis $\tilde\chi_n=\chi_n$ for $n>0$, and $\tilde\chi_0$ is the characteristic function of the whole space $\Z_d$ (while $\chi_0$ is the characteristic function of $d\Z_d=0X^\infty$). This difference comes exactly from two ways of defining $[0]_d$ that was mentioned earlier. If $[0]_d=0$, we obtain the version of basis used by Mahler, and defining $[0]_d=\varepsilon$ (the empty word) yields the basis used by Schikhof. This difference does not change much the results and the proofs and we will give formulations of some of our results for both bases. In particular, the decomposition~\eqref{eq:mahler} is transformed into
\begin{equation}
\label{eq:schikhof}
f(x)=\sum_{n\geq 0}\tilde B^f_n\tilde\chi_n(x),
\end{equation}
where Schikhof's versions of van der Put coefficients $\tilde B^f_n$ are computed as:
\begin{equation}
\label{eqn:B_n_schikhof}
\tilde B^f_n=\left\{
\begin{array}{ll}
f(0),&\text{if}\ \ n=0,\\
f(n)-f(n\_),&\text{if}\ \ n>0.
\end{array}
\right.
\end{equation}

Among all continuous functions $\Z_d\to\Z_d$ we are interested in those that define endomorphisms of $X^*$ (viewed as a tree). We will use the following useful characterization of these maps in terms of the coefficients of their van der Put series (which works for both versions of the van der Put basis). In the case of prime $d$ this easy fact is given in~\cite{anashin_ky:characterization_of_ergodic11}. The proof in general case is basically the same and we omit it.

\begin{theorem}
\label{thm:lipschitz}
A function $\Z_d\to\Z_d$ is 1-Lipschitz if and only if it can be represented as
\begin{equation}
\label{eq:reduced_decomposition}
f(x)=\sum_{n\geq 0}\vpc^f_nd^{\lfloor\log_dn\rfloor}\chi_n(x),
\end{equation}
where $\vpc^f_n\in\Z_d$ for all $n\geq 0$, and
\[\lfloor \log_dn\rfloor= (\text{the number of digits in the base-$d$ expansion of}\ n)-1.\]
\end{theorem}

We will call the coefficients $\vpc^f_n$ from Theorem~\ref{thm:lipschitz} the \emph{reduced van der Put coefficients}. It follows from equation~\eqref{eqn:B_n} that these coefficients are computed as
\begin{equation}
\label{eqn:lambda_n}
\vpc^f_n=B^f_n d^{-\lfloor\log_dn\rfloor}=\left\{
\begin{array}{ll}
f(n),&\text{if}\ \ 0\leq n<d,\\[2mm]
\frac{f(n)-f(n\_)}{d^{\lfloor\log_dn\rfloor}},&\text{if}\ \ n\geq d.
\end{array}
\right.
\end{equation}

For Schikhof's version of the van der Put basis equation~\eqref{eq:reduced_decomposition} has to be replaced with
\[f(x)=\sum_{n\geq 0}\tilde\vpc^f_nd^{\lfloor\log_dn\rfloor}\chi_n(x)\]
and corresponding reduced van der Put coefficients are computed as
\[\tilde\vpc^f_n=\tilde B^f_n d^{-\lfloor\log_dn\rfloor}=\left\{
\begin{array}{ll}
f(0),&\text{if}\ \ n=0,\\[2mm]
\frac{f(n)-f(n\_)}{d^{\lfloor\log_dn\rfloor}},&\text{if}\ \ n>0.
\end{array}
\right.
\]
In particular, $\tilde\vpc^f_n=\vpc^f_n$ for all $n\geq d$.

We note that since Schikhof's reduced van der Put coefficients $\tilde\vpc^f_n$ differ from the $\vpc^f_n$ only for $n<d$, the claim of Theorem~\ref{thm:characterization} clearly remains true for Schikhof's van der Put series as well.

\section{Automatic sequences}
\label{sec:automatic}
There are several equivalent ways to define $d$-automatic sequences. We will refer the reader to Allouche-Shallit's book~\cite{allouche:automatic_sequences03} for details. Informally, a sequence $(a_n)_{n\geq 0}$ is called $d$-automatic if one can compute $a_n$ by feeding a deterministic finite automaton with output (DFAO) the base-$d$ representation of $n$, and then applying the output mapping $\tau$ to the last state reached. We first recall the definition of the (Moore) DFAO and then give the formal definition of automatic sequences.

\begin{definition}
A \emph{deterministic finite automaton with output} (or a \emph{Moore} automaton) is defined to be a 6-tuple
\[\B=(Q,X,\delta,q_0,A,\tau)\]
where\\
\begin{itemize}
\item $Q$ is a finite \emph{set of states}
\item $X$ is the finite \emph{input alphabet}
\item $\delta\colon Q\times X\to Q$ is the \emph{transition function}
\item $q_0\in Q$ is the \emph{initial state}
\item $A$ is the \emph{output alphabet}
\item $\tau\colon Q\to A$ is the \emph{output function}
\end{itemize}
In the case when the input alphabet is $X=\{0,1,\ldots,d-1\}$ we will call the corresponding automaton a $d$-DFAO.
\end{definition}

Similarly to the case of Mealy automata, we extend the transition function $\delta$ to $\delta\colon Q\times X^*\to Q$. With this convention, a $d$-DFAO defines a function $f_M\colon X^*\to A$ by $f_M(w)=\tau(\delta(q_0,w))$.

Note, that Moore automata can also be viewed as transducers as well by recording the values of the output function at every state while reading the input word. This way each word over $X$ will be transformed into a word over $A$ of the same length. This model of calculations is equivalent to Mealy automata (in the more general case when the output alphabet is allowed to be different from the input alphabet) in the sense that for each Moore automaton there exists a Mealy automaton that defines the same transformation from $X^*$ to $A^*$ and vice versa (see~\cite{sholomov:osnovy_diskr80} for details).

Recall that for a word $w=x_0x_1\ldots x_n\in X^*$ we denote by $\overline w=x_0+x_1d+\cdots+x_nd^{n}\in\N_0$ the label of the closest to $w$ labelled vertex in $X^*$ along the unique path from $w$ to the root of the tree.

\begin{definition}[\cite{allouche:automatic_sequences03}]
\label{def:automaticity}
We say that a sequence $(a_n)_{n\geq 0}$ over a finite alphabet $A$ is \emph{$d$-automatic} if there exists a $d$-DFAO $\B=(Q,X,\delta,q_0,A,\tau)$ such that $a_n=\tau(\delta(q_0,w))$ for all $n\geq 0$ and $w\in X^*$ with $\overline w=n$.
\end{definition}

For us it will be more convenient to use the alternative characterization of automatic sequences (for the proof see for instance~\cite{allouche:automatic_sequences03}).

\begin{theorem}
\label{thm:automatic_kernel}
A sequence $(a_n)_{n\geq 0}$ over an alphabet $A$ is \emph{$d$-automatic} if and only if the collection of its subsequences of the form $\{(a_{j+n\cdot d^i})_{n\geq 0}\mid i\geq 0, 0\leq j<d^i\}$, called the \emph{$d$-kernel}, is finite.
\end{theorem}

We will recall the connection between the $d$-DFAO defining a $d$-automatic sequence $(a_n)_{n\geq 0}$ and the $d$-kernel of this sequence (see Theorem~6.6.2 in~\cite{allouche:automatic_sequences03}).
For that we define the section of a sequence $(a_n)_{n\geq 0}$ at a word $v$ over $X=\{0,1,\ldots,d-1\}$ recursively as follows.

\begin{definition}
\label{def:section}
Let $(a_n)_{n\geq 0}$ be a sequence over alphabet $A$. Its \emph{$d$-section $(a_n)_{n\geq 0}\bigl|_x$ at $x\in X=\{0,1,\ldots,d-1\}$} is a subsequence $(a_{x+nd})_{n\geq 0}$. For a word $v=x_1x_2\ldots x_k$ over $X$ we further define the \emph{$d$-section $(a_n)_{n\geq 0}\bigl|_v$ at $v$} to be either $(a_n)_{n\geq 0}$ itself if $v$ is the empty word or $(a_n)_{n\geq 0}\bigl|_{x_1}\bigl|_{x_2}\ldots\bigl|_{x_k}$ otherwise.
\end{definition}

We will often omit $d$ in the term $d$-section when $d$ is clear from the context. The $d$-kernel of a sequence consists exactly of $d$-sections and $d$-automaticity of a sequence can be reformulated as:

\begin{proposition}
\label{prop:automatic_sections}
A sequence $(a_n)_{n\geq 0}$ over an alphabet $A$ is $d$-automatic if and only if the set $\bigl\{(a_n)_{n\geq 0}\bigl|_v : v\in X^*\bigr\}$ is finite.
\end{proposition}

The subsequences involved in the definition of the $d$-kernel can be plotted on the $d$-ary rooted tree $X^*$, where the vertex $v\in X^*$ is labelled with the subsequence $(a_n)|_v$. For $d=2$ such a tree is shown in Figure~\ref{fig:subsequences_tree}.

\begin{figure}[h]
\begin{center}
\epsfig{file=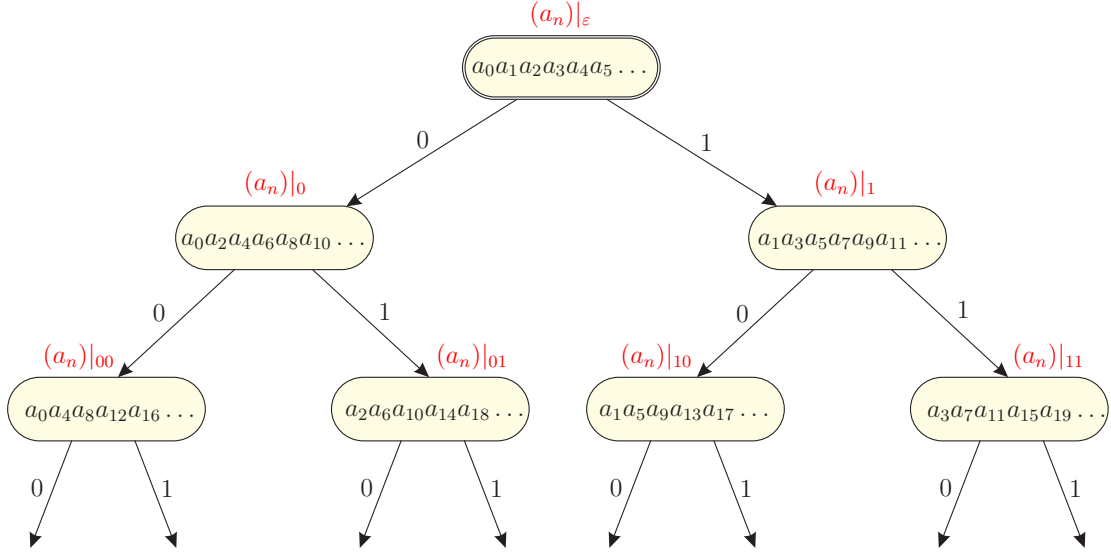,width=400pt}
\caption{The tree of subsequences of $(a_n)_{n\geq 0}$ constituting its $d$-kernel (for $d=2$)\label{fig:subsequences_tree}}
\end{center}
\end{figure}

A convenient way to represent sections of a sequence and understand $d$-automaticity is to put the terms of this sequence on a $d$-ary tree. Recall, that in the previous section we have constructed an embedding of $\N_0$ into $X^*$ via $n\mapsto [n]_d$. Under this embedding we will call the image of $n\in\mathbb N\cup\{0\}$ the vertex $n$ of $X^*$.

\begin{definition}
The $d$-\emph{portrait} of a sequence $(a_n)_{n\geq0}$ over an alphabet $A$ is a $d$-ary rooted tree $X^*$ where the vertex $n$ is labelled by $a_n$ and other vertices are unlabelled.
\end{definition}

In other words, we label each vertex $v=x_0x_1\ldots x_k$ with $x_k\neq 0$ or $v=0$ by $a_{\overline{v}}=a_{x_0+x_1d+\cdots+x_kd^k}$. For example, Figure~\ref{fig:binary_portrait_an} represents the 2-portrait of the sequence $(a_n)_{n\geq0}$.

\begin{figure}[h]
\begin{center}
\epsfig{file=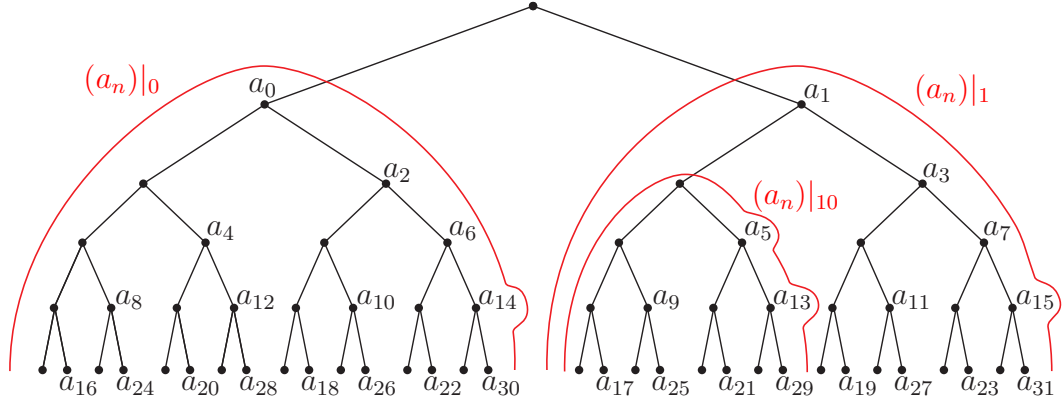, width=400pt}
\caption{The 2-portrait of the sequence $(a_n)_{n\geq0}$\label{fig:binary_portrait_an}}
\end{center}
\end{figure}

To simplify the exposition we will write simply \emph{portrait} for $d$-portrait when the value of $d$ is clear from the context. In particular, unless otherwise stated, $X$ will denote an alphabet $\{0,1,\ldots,d-1\}$ of cardinality $d$ and a portrait will mean a $d$-portrait.

There is a simple connection between the portrait of a sequence and the portrait of its section at vertex $v\in X^*$ that takes into account that the subtree $vX^*$ of $X^*$ hanging down from vertex $v$ is canonically isomorphic to $X^*$ itself via $vu\leftrightarrow u$ for each $u\in X^*$.

\begin{proposition}
\label{prop:portrait}
For a sequence $(a_n)_{n\geq 0}$ over an alphabet $A$ with a portrait $P$ and a vertex $v=x_0x_1\ldots x_k$, $k\geq0$ of $X^*$ the portrait of the section $(a_n)_{n\geq 0}\bigl|_v$ is obtained from the portrait of $(a_n)_{n\geq 0}$ by taking the (labelled) subtree of $P$ hanging down from vertex $v$, removing, if $v$ ends with $x_k\neq 0$ and $k>0$, the label at its root vertex, and labelling the vertex 0 by $a_{\overline{v}}=a_{x_0+x_1d+\cdots+x_kd^k}$, which is the label of the closest to $v0$ labelled vertex in $P$ on the unique path connecting $v0$ to the root.
\end{proposition}

The proof of the above proposition follows immediately from the definitions of portrait and section.

In other words, as shown in Figure~\ref{fig:binary_portrait_an}, you can see the portrait of a section of a sequence $(a_n)_{n\geq 0}$ at vertex $v\in X^*$ just by looking at the subtree hanging down in the portrait of $(a_n)_{n\geq 0}$ from vertex $v$ (modulo minor technical issue of labelling the vertex $0$ of this subtree and possibly removing the label of the root vertex). Therefore, a sequence is automatic if and only if its portrait has finite number of ``subportraits'' hanging down from its vertices. This way of interpreting automaticity now corresponds naturally to the condition of an automaton endomorphism being finite state.

Note, that the formulation of the previous proposition would be simpler had we defined portraits by labelling each vertex $v=x_0x_1\ldots x_k$ of the tree by $a_{x_0+x_1d+\cdots x_kd^k}$ instead of only numbered ones, but we have intentionally opted not to do that to simplify notations in the next section.

Now it is easy to see that the $d$-DFAO defining a $d$-automatic sequence $(a_n)_{n\geq0}$ over an alphabet $A$ with the $d$-kernel $K$ can be built as follows.

\begin{proposition}
\label{prop:moore_automaton}
Suppose $(a_n)_{n\geq0}$ is a $d$-automatic sequence over an alphabet $A$ with the $d$-kernel $K$. Then a $d$-DFAO $\B=(K,X,\delta,q_0,A,\tau)$, where:
\begin{equation}
\label{eqn:moore_transitions}
\begin{array}{rcl}
\delta((a_n)_{n\geq0}|_v, x)&=&(a_n)_{n\geq0}|_{vx},\\
\tau((a_n)_{n\geq0}|_v)&=&a_{\overline{v}}\ (\text{the first term of the sequence}\ (a_n)_{n\geq0}|_v),\\
q_0&=&(a_n)_{n\geq0}|_{\varepsilon}=(a_n)_{n\geq0}.
\end{array}
\end{equation}
defines the sequence $(a_n)_{n\geq0}$.
\end{proposition}

Informally, we build the automaton $M$ by following the edges of the tree $X^*$ from the root, labelling these edges by corresponding elements of $X$, and identifying the vertices that correspond to the same sections of $(a_n)_{n\geq0}$ into one state of $M$ that is labelled by the $0$-th term of the corresponding section.

\section{Portraits of sequences of reduced van der Put coefficients and their sections}
\label{sec:portraits}

It turns out that there is a natural relation between the (portraits of the sequences of) reduced van der Put coefficients of an endomorphism $g$ and of its sections. Denote by $\sigma\colon\Z_d\to\Z_d$ the map $\sigma(a)=\frac{a-(a\mmod d)}d$. This map corresponds to the shift map on $\Z_d$ that deletes the first letter of $a$. I.e., if $a=x_0x_1x_2\ldots\in\Z_d$, then $\sigma(a)=x_1x_2x_3\ldots\in\Z_d$.

\begin{theorem}
\label{thm:van_der_put_sections}
Suppose $g\in\End X^*$ has sections $g|_x$, $x=0,1,\ldots,n-1$ at the vertices of the first level of $X^*$. Then the reduced van der Put coefficients $b_n^{g|_x}$ of the section $g|_x$ satisfy:
\begin{equation}
\label{eqn:lambda_section}
\vpc_n^{g|_x}=\left\{\begin{array}{ll}
\sigma(\vpc_x^g),& n=0,\\
\vpc_{x+nd}^g+\sigma(\vpc_x^g),& 0<n<d,\\
\vpc_{x+nd}^g,& n\geq d
\end{array}
\right.
\end{equation}
where for $\vpc\in\Z_d$ we denote by $\sigma(\vpc)=\frac{\vpc-(\vpc\mmod d)}d$ the shift map on $\Z_d$.
\end{theorem}

\begin{proof}
First we consider the case $n=0$. By equation~\eqref{eqn:lambda_n} the reduced van der Put coefficients are computed as follows
\[\vpc_0^{g|_x}=g|_x(0^\infty)=\frac{g(x0^\infty)-(g(x0^\infty)\mmod d)}d=\sigma(g(x0^\infty))=\sigma(\vpc_x^g).\]
Similarly for $0<n<d$ we obtain
\begin{multline*}\vpc_n^{g|_x}=g|_x(n0^\infty)=\frac{g(xn0^\infty)-(g(xn0^\infty)\mmod d)}d\\
=\frac{g(xn0^\infty)-(g(x0^\infty)\mmod d)}d+\frac{g(x0^\infty)-(g(xn0^\infty)\mmod d)}d\\
=\vpc_{x+nd}^g+\frac{g(x0^\infty)-(g(x0^\infty)\mmod d)}d=\vpc_{x+nd}^g+\sigma(\vpc_x^g).
\end{multline*}
Finally, for $n>d$ we derive
\begin{multline*}\vpc_n^{g|_x}=d^{-\lfloor\log_dn\rfloor}\bigl(g|_x([n]_d0^\infty)-g|_x([n\_]_d0^\infty)\bigr)\\
=d^{-\lfloor\log_dn\rfloor}\left(\frac{g(x[n]_d0^\infty)-(g(x[n]_d0^\infty)\mmod d)}d-\frac{g(x[n\_]_d0^\infty)-(g(x[n\_]_d0^\infty)\mmod d)}d\right)\\
=d^{-\lfloor\log_dn\rfloor-1}\bigl(g(x[n]_d0^\infty)-g(x[n\_]_d0^\infty)\bigr)\\
=d^{-\lfloor\log_d(x+nd)\rfloor}\bigl(g([x+nd]_d0^\infty)-g([(x+nd)\_]_d0^\infty)\bigr)=\vpc_{x+nd}^g,
\end{multline*}
where in the last line we used that for $x<d$ we have $x+(n\_)d=(x+nd)\_$ and
\[\lfloor\log_d(n)\rfloor+1=\lfloor\log_d(n)+1\rfloor=\lfloor\log_d(nd)\rfloor=\lfloor\log_d(x+nd)\rfloor.\]
\end{proof}

In the case of Schikhof's version of van der Put basis we can similarly prove the following.

\begin{theorem}
\label{thm:van_der_put_sections_schikhof}
Suppose $g\in\End X^*$ has sections $g|_x$, $x=0,1,\ldots,n-1$ at the vertices of the first level of $X^*$. Then the reduced van der Put coefficients with respect to Schikhof's version of van der Put basis of the section $g|_x$ satisfy:
\begin{equation}
\label{eqn:lambda_section_schikhof}
\tilde\vpc_n^{g|_x}=\left\{\begin{array}{ll}
\sigma(\tilde\vpc_0^g),& n=0, x=0\\
\sigma(\tilde\vpc_{x}^g+\vpc_0^g),& n=0, 0<x<d,\\
\tilde\vpc_{x+nd}^g,& n>0.
\end{array}
\right.
\end{equation}
\end{theorem}

There is a more visual way to state the third case in equation~\eqref{eqn:lambda_section} using the $\overline{\phantom{a}}$ notation.
\begin{corollary}
\label{cor:sections1}
Let $x_0x_1\ldots x_k\in X^*$ be a word of length $k+1\geq 3$ with $x_k\neq 0$. Then
\[\vpc^g_{\overline{x_0x_1\ldots x_k}}=\vpc^{g|_{x_0}}_{\overline{x_1\ldots x_k}}.\]
\end{corollary}

\begin{proof}
Follows from~\eqref{eqn:lambda_section} and the fact that if $\overline{x_1x_2\ldots x_k}=n$, then $\overline{x_0x_1x_2\ldots x_k}=x_0+nd$.
\end{proof}

The next corollary will be used in calculations in Section~\ref{sec:examples}.

\begin{corollary}
\label{cor:sections2}
Let $v,w\in X^*$ with $w$ of length at least 2 and ending in a nonzero element of $X$.
Then \[\vpc^g_{\overline{vw}}=\vpc^{g|_v}_{\overline{w}}.\]
\end{corollary}

\begin{proof}
When $v$ is the empty word the claim is trivial. The general case now follows by induction on $|v|$ from Corollary~\ref{cor:sections1} as for each $x\in X$ we have
\[\vpc^g_{\overline{xvw}}=\vpc^{g|_x}_{\overline{vw}}=\vpc^{(g|_x)|_v}_{\overline{w}}=\vpc^{g|_{xv}}_{\overline{w}}.\]
\end{proof}

\begin{corollary}
\label{cor:correspondence}
Let $g\in\End(X^*)$ be an endomorphism of $X^*$ and $v\in X^*$ be arbitrary vertex. Then the sequences $(\vpc_n^{g|_v})_{n\geq 0}$ and $(\vpc_n^{g})_{n\geq 0}|_v$ coincide starting from term $d$.
\end{corollary}

\begin{proof}
For any $n\geq d$ we have that $[n]_d=xw$ for some $x\in X$ and $w\in X^*$ of length at least 1 that ends with a non-zero element of $X$. So we have by Corollary~\ref{cor:sections2}
\[\vpc_n^{g|_v}=\vpc_{\overline{xw}}^{g|_v}=\vpc_{\overline{vxw}}^{g}.\]
But $\vpc_{\overline{vxw}}^{g}$ is exactly the term of the sequence $(\vpc_n^{g})_{n\geq 0}|_v$ with index $n=\overline{xw}$.
\end{proof}

Now, taking into account Proposition~\ref{prop:portrait}, there is a geometric way to look at the previous theorem. Namely, the third subcase in equation~\eqref{eqn:lambda_section} yields the following proposition.

\begin{corollary}
\label{cor:correspondence2}
Let $v\in X^*$ be arbitrary vertex of $X^*$. The labels of the portrait of the sequence $(\vpc_n^{g|_v})_{n\geq 0}$ coincide at levels 2 and below with the corresponding labels of the restriction of the portrait of $(\vpc_n^{g})_{n\geq 0}$ to the subtree hanging down from vertex $v\in X$.
\end{corollary}

We illustrate by Figure~\ref{fig:portrait_correspondence} this fact for $v=x\in X$ of length one, where the portraits of $(\vpc_n^{g|_0})_{n\geq 0}$ and $(\vpc_n^{g|_1})_{n\geq 0}$ are drawn on the left and right subtrees of the portrait of $(\vpc_n^{g})_{n\geq 0}$. Figure~\ref{fig:portrait_correspondence} asserts that the labels of the portraits of sections coincide with the labels of the portrait of $(\vpc_n^{g})_{n\geq 0}$ below the dashed line. The first two subcases of~\eqref{eqn:lambda_section} give labels of the portraits of $(\vpc_n^{g|_x})_{n\geq 0}, x\in X$ on the first level.

\begin{figure}[h]
\begin{center}
~\hspace{-4cm}\epsfig{file=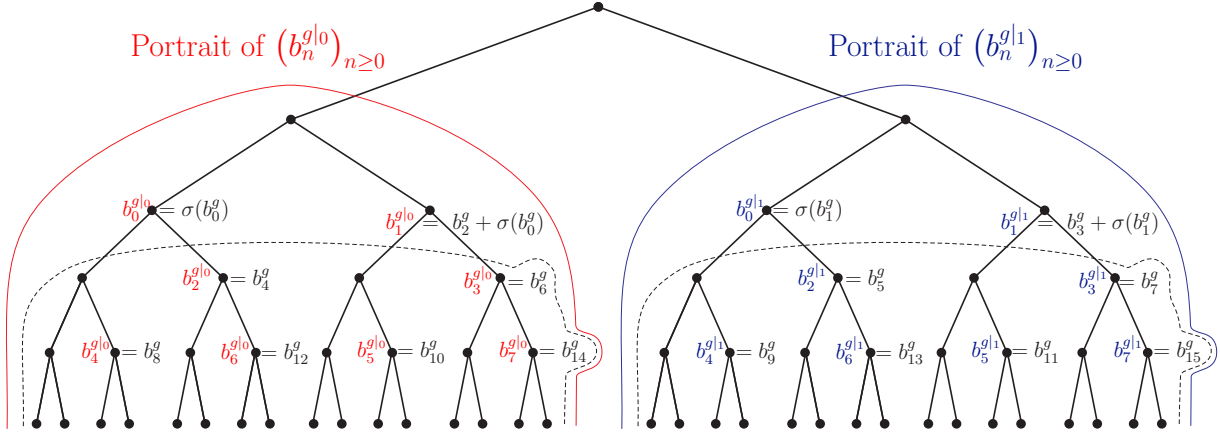, width=350pt}
\caption{Correspondence between portraits of $(\vpc_n^{g|_x})_{n\geq 0}$ and $(\vpc_n^{g})_{n\geq 0}$\label{fig:portrait_correspondence}}
\end{center}
\end{figure}

\section{Proof of the Theorem~\ref{thm:characterization}}
\label{sec:proof}

In this section we will prove Theorem~\ref{thm:characterization}. In the arguments below we will work with eventually periodic elements of $\Z_d$, i.e. elements of the form $a_0+a_1d+a_2d^2+\cdots$ with eventually periodic sequence $(a_i)_{i\geq 0}$ of coefficients. As shown in~\cite[Theorem~4.2.4]{goresky_k:alg_shift_register_sequences12}, this set of $d$-adic integers can be identified with the subset $\Z_{d,0}$ of $\Q$ consisting of all rational numbers $a/b\in\Q$ such that $b$ is relatively prime to $d$. Algebraically it can be defined as $\Z_{d,0}=D^{-1}\Z$, where $D$ is the multiplicative set $\{b\in\Z\colon \gcd(b,d)=1\}$. We will denote the corresponding inclusion $\Z_{d,0}\hookrightarrow\Z_d$ by $\psi$. The following inclusions then take place.
\[\begin{array}{ccccccc}
\Z&\subset&\Z_{d,0}&\stackrel{\psi}{\hookrightarrow}&\Z_d&\subset&\Q_d\\
&&\cap&&\\
&&\Q&&
\end{array}
\]

We  will  not  need the  definition  of $\psi$  which  can  be  constructed using Lemma~4.2.2 in~\cite{goresky_k:alg_shift_register_sequences12}, but rather will need  the  definition  of  $\psi^{-1}\colon \psi(\Z_{d,0})\to\Z_{d,0}$. The  map  is defined  as follows. Suppose $uv^\infty\in\Z_d$ is an arbitrary eventually periodic element for some $u,v\in X^*$. Then we define
\[\psi^{-1}(uv^\infty)=\overline u+\frac{\overline v\cdot d^{|u|}}{1-d^{|v|}}\in\Z_{d,0}.\]

\begin{lemma}
\label{lem:periodic}
The preimage under $\psi$ of the set $\{v^\infty\colon v\in X^m\}$ of all periodic elements of $\Z_d$ with period of length dividing $m\geq 0$ is the set
\[P^{0,m}=\left\{\frac j{1-d^m}\colon 0\leq j<d^m\right\}\]
which is the subset of the interval $[-1,0]\subset\R$.
\end{lemma}

\begin{proof}
It follows from the definition of $\psi^{-1}$ that
\[\psi^{-1}(v^\infty)= \frac{\overline v}{1-d^{|v|}}.\]
Recall, that for $v=x_0x_1\ldots x_{m-1}$ we have $\overline v=x_0+x_1d+\ldots x_{m-1}d^{m-1}$. This implies that $0\leq\overline v\leq d^{m}-1$ and, henceforth, $-1\leq\psi^{-1}(v^\infty)\leq 0$. Moreover, as $v$ runs over all words in $X^m$, $\overline v$ runs over all integer numbers from $0$ to $d^m-1$ as we simply list the $d$-ary expansions of all these numbers.
\end{proof}

\begin{lemma}
\label{lem:preperiodic}
The preimage under $\psi$ of the set $\{uv^\infty\colon u\in X^l, v\in X^m\}$ of all eventually periodic elements of $\Z_d$ with preperiod of length at most $l\geq 0$ and period of length dividing $m\geq 1$ is the set
\[P^{l,m}=\left\{i+\frac{j\cdot d^l}{1-d^m}\colon 0\leq i<d^l, 0\leq j<d^m\right\}\]
which is the subset of the interval $[-d^l,d^l-1]\subset\R$.
\end{lemma}

\begin{proof}
We have
\begin{multline*}
\psi^{-1}(\{uv^\infty\colon u\in X^l,v\in X^m\})=\{\psi^{-1}(uv^\infty)\colon u\in X^l,v\in X^m\}\\
=\{\psi^{-1}(u0^\infty)+\psi^{-1}(0^lv^\infty)\colon u\in X^l,v\in X^m\}=\{\overline u+d^l\cdot\psi^{-1}(v^\infty)\colon  u\in X^l,v\in X^m\}\\
=\{i+d^l\psi^{-1}(v^\infty)\colon  0\leq i<d^l,v\in X^m\}=\bigcup_{0\leq i<d^l}(i+d^lP^{0,m}).
\end{multline*}
The set $d^lP^{0,m}$ by Lemma~\ref{lem:periodic} is a subset of $[-d^l,0]$. Therefore, since $P^{l,m}\subset\Q$ is obtained as the union of all shifts of $d^lP^{0,m}$ by all integers $0\leq i<d^l$, we obtain that $P^{l,m}\subset[-d^l,d^l-1]$.
\end{proof}

In order to obtain the condition of finiteness of Mealy automata in the proof of Theorem~\ref{thm:characterization} we will also need the following technical lemma. Define a sequence of subsets $A_i^{l,m}$ recursively by $A_0^{l,m}=P^{l,m}$, and

\begin{equation}
\label{eqn:A}
A_{i+1}^{l,m}=\sigma(A_i^{l,m})+P^{l,m}.
\end{equation}

$A^{l,m}_{i}$ will be used later to describe possible sets of states of automaton defined by a transformation with a given automatic sequence of reduced van der Put coefficients.

\begin{lemma}
\label{lem:finite}
The set $A^{l,m}=\bigcup_{i\geq 0}A^{l,m}_i$ is finite.
\end{lemma}

\begin{proof}
First, we remark that the denominators of fractions in $A^{l,m}_i$ are divisors of $d^m-1$. Therefore, it will be enough to prove by induction on $i$ that $A^{l,m}_i\subset[-z,z]$ for $z=\frac{d^{l+1}+d-1}{d-1}$. For $i=0$ the statement is true since $A_0^{l,m}\subset[-d^l,d^l-1]$ by Lemma~\ref{lem:periodic}, and
\[z=\frac{d^{l+1}+d-1}{d-1}>\frac{d^{l+1}+d-1}{d}=d^l+\frac{d-1}d>d^l.\]
Assume that the statement is true for a given $i\geq 0$. Any element of $A_{i+1}^{l,m}$ is equal to $\sigma(x)+b$ for some $x\in A_{i+1}^{l,m}\subset[-z,z]$ and $b\in P^{l,m}\subset[-d^l,d^l-1]$. Since $\sigma(x)=\frac{x-x\mmod d}d$, we immediately obtain
\begin{multline*}\sigma(x)+b\leq\frac xd+b\leq \frac zd+d^l-1=\frac{(d^{l+1}+d-1)}{d(d-1)}+d^{l}-1\\=\frac{d^l+1+(d^{l}-1)(d-1)}{d-1}=\frac{d^{l+1}-d+2}{d-1}<z.
\end{multline*}
For the lower bound we obtain
\begin{multline*}\sigma(x)+b\geq\frac {x-d+1}d+b\geq \frac{-z-d+1}d-d^l=\frac{-\frac{d^{l+1}+d-1}{d-1}-d+1}{d}-d^{l}=\\
\frac{-d^l-d+1}{d-1}-d^l=-\frac{d^{l+1}+d-1}{d-1}=-z.
\end{multline*}
\end{proof}

We are ready to proceed to the main result of this section.

\begin{proof}[Proof of Theorem~\ref{thm:characterization}]
First, assume that $g\in\End(X^*)$ is defined by a finite Mealy automaton $\A$ with the set of states $Q$. In order to prove that $(\vpc^g_n)_{n\geq 0}$ is automatic, by Proposition~\ref{prop:automatic_sections} we need to show that it has finitely many sections at vertices of $X^*$.

Assume that $v\in X^*$ is of length at least 2, $v=v'xy$ for some $v'\in X^*$ and $x,y\in X$. Then the section $(\vpc^g_n)|_v$ is a sequence that can be completely identified by a pair
\begin{equation}
\label{eq:pair}
\bigl(\vpc^g_{\overline v},\sigma((\vpc^g_n)|_v)\bigr),
\end{equation}
where $\vpc^g_{\overline v}$ is its zero term, and $\sigma((\vpc^g_n)|_v)$ is the subsequence made of all other terms.

Since by Corollary~\ref{cor:correspondence2}%XXX Add detailed explanation
\[\sigma((\vpc^g_n)|_v)=\sigma((\vpc^g_n)|_{v'xy})=\sigma((\vpc^{g|_{v'}}_n)|_{xy}),\]
the number of possible choices for the second component in~\eqref{eq:pair} is bounded above by $|Q|\cdot|X|^2$ (as we have $|Q|$ choices for $g|_{v'}$ and $|X|^2$ choices for $xy\in X^2$).

Further, if $\overline v<d$ (i.e., $v=z0^k$ for some $z\in X$) then the number of choices for the first component $\vpc^g_{\overline v}$ of~\eqref{eq:pair} is bounded above by $|Q|\cdot|X|$. Otherwise, $v'=v''x'y'$ for some $v''\in X^*$, $x',y'\in X$ with $y'\neq0$. In this case, $\vpc^g_{\overline v}=\vpc^{g|_{v''}}_{\overline{x'y'}}$, so the number of possible choices for $\vpc^g_{\overline v}$ is again bounded above by $|Q|\cdot|X|^2$. Thus, the sequence $(\vpc^g_n)_{n\geq0}$ has finitely many sections.

To prove the first condition asserting that all $\vpc^g_n$ are in $\Z_d\cap\Q$, or, equivalently, eventually periodic, it is enough to mention that by equation~\eqref{eqn:lambda_n} $\vpc^g_n$ must be eventually periodic for $n\geq d$ as a shifted difference of two eventually periodic words $g(n)=g([n]_d0^\infty)$ and $g(n\_)=g([n\_]_d0^\infty)$. The latter two words are eventually periodic as the images of eventually periodic words $[n]_d0^\infty$ and $[n\_]_d0^\infty$ under a finite automaton transformation. Similar even works for $n<d$, in which case there is no need to take a difference.

Now we prove the converse direction. Assume that for $g\in\End(X^*)$ of $X^*$ the sequence $(\vpc^g_n)_{n\geq0}$ is automatic and consists of eventually periodic elements of $\Z_d$. Then automaticity implies that $\{\vpc^g_n : n\geq0\}$ is finite as a set. Let $l$ be the maximal length among preperiods of all $\vpc^g_n$, and let $m$ be the least common multiple of the lengths of all periods of $\vpc^g_n$. Then clearly $\vpc^g_n\in\psi(P^{l,m})$ for all $n\geq0$ by definition of $P^{l,m}$ given in Lemma~\ref{lem:preperiodic}.

Our aim is to show that the set $\{g|_v : v\in X^*\}$ is finite. We will show that there is only finitely many portraits $(\vpc^{g|_v}_n)_{n\geq0}$. Let $v\in X^*$. By Corollary~\ref{cor:correspondence2} the part of the portrait of $(\vpc^{g|_v}_n)_{n\geq0}$ below level one coincides with the part below level one of the restriction of the portrait of $(\vpc^{g}_n)_{n\geq0}$ on the subtree hanging down from vertex $v$. But according to Propositions~\ref{prop:automatic_sections} and~\ref{prop:portrait}, since $(\vpc^{g}_n)_{n\geq0}$ is automatic, there is only finite number of such restrictions as the set of all sections $\{(\vpc^{g}_n)|_v: v\in X^*\}$ is finite.

Hence, we only need to check that there is a finite number of choices for the van der Put coefficients of the first level of $g|_v$ for $v\in X^*$. To do that we will prove by induction on $|v|$ that $\vpc^{g|_v}_i\in \psi(A^{l,m}_{|v|})$ for $0\leq i<d$. The claim is trivial for $|v|=0$ by definition of $A^{l,m}$ and the choice of $l$ and $m$. Assume that the claim is true for all words $v$ of length $k$, and let $vx$ be a word of length $k+1$ for some $x\in X$. Then by assumption $\vpc^{g|_v}_x\in \psi(A^{l,m}_{|v|})$, and additionally $\vpc^{g|_v}_{x+d\cdot i}\in P^{l,m}$ for $1\leq i<d$ as these coefficients of $g$ on the second level of its portrait coincide with corresponding coefficients of $g$.  Now by Theorem~\ref{thm:van_der_put_sections} we obtain
\begin{equation*}
%\label{eqn:lambda_section}
\vpc_i^{g|_{vx}}=\vpc_i^{(g|_v)|_x}=
\left\{\begin{array}{ll}
\sigma(\vpc_x^{g|_v}),& i=0,\\
\vpc_{x+d\cdot i}^{g|_v}+\sigma(\vpc_x^{g|_v}),& 0<i<d,
\end{array}
\right.
\end{equation*}
In both cases we get that $\vpc_i^{g|_{vx}}\in A^{l,m}_{|vx|}$ by definition of $A^{l,m}_{|vx|}$ from~\eqref{eqn:A}. Finally, Lemma~\ref{lem:finite} now guarantees that $g$ has finitely many sections and completes the proof.
\end{proof}

\section{Mealy and Moore Automata Associated to an Endomorphism of $X^*$}
\label{sec:relation}

The above proof of Theorem~\ref{thm:characterization} allows us to build algorithms that construct the Moore automaton of the automatic sequence of reduced van der Put coefficients of a transformation of $\Z_d$ defined by a finite state Mealy automaton, and vice versa.

We start from constructing the Moore automaton generating the sequence of reduced van der Put coefficients of an endomorphism $g$ from the finite Mealy automaton defining $g$.

\begin{theorem}
\label{thm:covering}
Let $g\in\End(X^*)$ be an endomorphism of $X^*$ defined by a finite initial Mealy automaton $\A$ with the set of states $Q_{\A}=\{g|_v\colon v\in X^*\}$. Let also $(\vpc_n^g)_{n\geq0}$ be the sequence of reduced van der Put coefficients of the map $\Z_d\to\Z_d$ induced by $g$. Then the Moore automaton $\B=(Q_{\B},X,\delta,q,\Z_d,\tau)$, where
\begin{itemize}
\item the set of states is $Q_{\B}=\left\{\bigl(g|_v,(\vpc^g_{\overline{vy}})_{y\in X}\bigr)\colon v\in X^*\right\}$,
\item the transition and output functions are
\begin{equation}
\label{eqn:states_labels}
\begin{array}{rcl}
\delta\left(\bigl(g|_v,(\vpc^g_{\overline{vy}})_{y\in X}\bigr),x\right)&=&\bigl(g|_{vx},(\vpc^g_{\overline{vxy}})_{y\in X}\bigr),\\
\tau\left(\bigl(g|_v,(\vpc^g_{\overline{vy}})_{y\in X}\bigr)\right)&=&\vpc^g_{\overline{v}},
\end{array}
\end{equation}
\item the initial state is $q=\bigl(g,(\vpc^g_{\overline{y}})_{y\in X}\bigr)$
\end{itemize}
is finite and generates the sequence $(\vpc_n^g)_{n\geq0}$.
\end{theorem}

\begin{proof}
According to Proposition~\ref{prop:moore_automaton} one can construct an automaton $\B'$ generating $(\vpc_n^g)_{n\geq0}$ as follows. The states of $\B'$ are the sections of $(\vpc_n^g)_{n\geq0}$ at the vertices of $X^*$ (i.e., the $d$-kernel of $(\vpc_n^g)_{n\geq0}$) with the initial state being the whole sequence $(\vpc_n^g)_{n\geq0}$, and transition and output functions defined by~\eqref{eqn:moore_transitions}. Let $v\in X^*$ be an arbitrary vertex. By Corollary~\ref{cor:correspondence2} the labels of the portrait of $(\vpc_n^g)_{n\geq0}|_v$ at level $2$ and below coincide with the corresponding labels of the portrait of $(\vpc_n^{g|_v})_{n\geq0}$. Therefore, each state $(\vpc_n^g)_{n\geq0}|_v$ of $\B'$ can be completely defined by a pair, called the \emph{label} of this state:
\begin{equation}
\label{eqn:states_labels2}
l((\vpc_n^g)_{n\geq0}|_v)=\bigl(g|_v,(\vpc^g_{\overline{vy}})_{y\in X}\bigr),
\end{equation}
where $(\vpc^g_{\overline{vy}})_{y\in X}$ is the $d$-tuple of the first $d$ terms of $(\vpc_n^g)_{n\geq0}|_v$ that corresponds to the labels of the first level of the portrait of this sequence. The first component of this pair defines the terms of $(\vpc_n^{g|_v})_{n\geq0}$ at level 2 and below, and the second component consists of terms of the first level. It is possible that different labels will define the same state of $\B'$, but clearly the automaton $\B$ from the statement of the theorem also generates $(\vpc_n^g)_{n\geq0}$ since its minimization coincides with $\B'$.%XXX reference for minimization
Indeed, the set of states of $\B$ is the set of labels of states of $\B'$ and the transitions in $\B$ are obtained from the transitions in $\B'$ defined in Proposition~\ref{prop:moore_automaton}, and the definition of labels.

Finally, the finiteness of $Q_{\B}$ follows from our proof of Theorem~\ref{thm:characterization} since the set $\{g|_v\colon v\in X^*\}$ (coinciding with $Q_{\A}$) is finite, and the set $\{\vpc^g_{\overline{vy}}\colon v\in X^*, y\in X\}$ is a subset of a finite set $\{\vpc^{g'}_{\overline{w}}\colon g'\in Q_{\A}, w\in X\cup X^2\}$.
\end{proof}

For the algorithmic procedure that, given a finite state $g\in\End(X^*)$, constructs a Moore automaton generating the sequence $(\vpc^g_n)_{n\geq0}$ of its reduced van der Put coefficients, we need the following lemma.

\begin{lemma}
\label{lem:vpc_algorithm}
Given a finite state endomorphism $g\in G$ acting on $X^*$ with $|X|=d$, its first $d^2$ reduced van der Put coefficients $\vpc_{\overline{v}}^g$, $v\in X^*$ of length at most 2, are eventually periodic elements of $\Z_d$ that can be algorithmically computed.
\end{lemma}

\begin{proof}
Suppose $g$ has $q$ states. If $\overline v=i<d$, then by definition $\vpc_i^g=g(i0^\infty)$ is the image of an eventually periodic word under a finite automaton transformation. Thus, it is also eventually periodic with the period of length at most $q$ and the preperiod of length at most $q+1$. Clearly both period and preperiod can be computed effectively. Further, if $d\leq \overline v<d^2$, then $v=xy$ for $x,y\in X$ with $y\neq 0$. In this case $\vpc^g_{\overline v}=\frac{g(xy0^\infty)-g(x0^\infty)}{d}$ is eventually periodic as a shifted difference of two eventually periodic words $g(xy0^\infty)$ and $g(x0^\infty)$. The latter two words are eventually periodic as the images of eventually periodic words $xy0^\infty$ and $x0^\infty$ under a finite automaton transformation that can be effectively computed.
\end{proof}

\begin{algorithm}[Construction of Moore automaton from Mealy automaton]
\label{alg:mealy_to_moore}
Suppose an endomorphism $g$ of $X^*$ is defined by a finite state Mealy automaton $\A$ with the set of states $Q_{\A}$. To construct a Moore automaton $\B$ defining the sequence of reduced van der Put coefficients $(\vpc^g_n)_{n\geq 0}$ complete the following steps.
\begin{enumerate}
\item[Step 1.] Compute $\vpc^{g'}_{\overline{w}}$ for each $g'\in Q_{\A}$ and $w\in X\cup X^2$.
\item[Step 2.] Start building the set of states of $\B$ from its initial state $q=\bigl(g,(\vpc^g_{\overline{y}})_{y\in X}\bigr)$ with $\tau(q)=\vpc^g_0$. Define $Q_0=\{q\}$.
\item[Step 3.] To build $Q_{i+1}$ from $Q_i$ for $i\geq 1$ start from the empty set and for each state $q=\bigl(g|_v,(\vpc^{g}_{\overline{vy}})_{y\in X}\bigr)\in Q_i$ and each $x\in X$ add the state $q_x=\bigl(g|_{vx},(\vpc^{g}_{\overline{vxy}})_{y\in X}\bigr)$ to $Q_{i+1}$ unless it belongs to $Q_{j}$ for some $j\leq i$ or is already in $Q_{i+1}$. Use Corollary~\ref{cor:sections2} to identify $\vpc^{g}_{\overline{vxy}}$ with one of the elements computed in Step 1. Extend the transition function by $\delta(q,x)=q_x$ and the output function by $\tau(q_x)=\vpc^{g}_{\overline{vx}}$.
\item[Step 4.] Repeat Step 3 until $Q_{i+1}=\emptyset$.
\item[Step 5.] The set of state of the Moore automaton $\B$ is $\cup_{i\geq0}Q_i$, where the transition and output functions are defined in Step 3.
\end{enumerate}
\end{algorithm}

A particular connection between the constructed Moore automaton $\B$ and the original Mealy automaton $\A$ can be seen at the level of underlying oriented graphs as explained below.

\begin{definition}
For a Mealy automaton $\A=(Q,X,\delta,\lambda)$ we define its \emph{underlying oriented graph} $\Gamma(\A)$ to be the oriented labeled graph whose set of vertices is the set $Q$ of states of $\A$ and whose edges correspond to the transitions of $\A$ and are labeled by the input letters of the corresponding transitions. I.e., there is an oriented edge from $q\in Q$ to $q'\in Q$ labeled by $x\in X$ if and only if $\delta(q,x)=q'$.
\end{definition}

In other words, the underlying oriented graph of a Mealy automaton $\A$ can be obtained from the Moore diagram of $\A$ by removing second components of edge labels. For example, Figure~\ref{fig:underlying} depicts the underlying graph of a Mealy automaton from Figure~\ref{fig:lamp_aut} generating the lamplighter group $\LL$. Similarly, we construct underlying oriented graph of a Moore automaton.

\begin{figure}[h]
\begin{center}
\epsfig{file=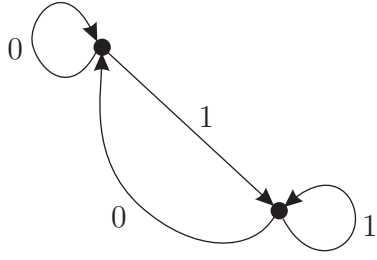}
\caption{Underlying graph for Mealy automaton from Figure~\ref{fig:lamp_aut} and Moore automaton from Figure~\ref{fig:thue_morse_moore}\label{fig:underlying}}
\end{center}
\end{figure}

\begin{definition}
For a Moore automaton $\B=(Q,X,\delta,q_0,A,\tau)$ we define its \emph{underlying oriented graph} $\Gamma(\B)$ to be the oriented labeled graph whose set of vertices is the set $Q$ of states of $\B$ and whose edges correspond to transitions of $\B$ and are labeled by the input letters of the corresponding transitions. I.e., there is an oriented edge from $q\in Q$ to $q'\in Q$ labeled by $x\in X$ if and only if $\delta(q,x)=q'$.
\end{definition}

Figure~\ref{fig:underlying} depicts also the underlying graph of a Moore automaton from Figure~\ref{fig:thue_morse_moore} generating the Thue-Morse sequence.

We finally define a \emph{covering} of such oriented labeled graphs to be a surjective (both on vertices and edges) graph homomorphism that preserves the labels of the edges.

\begin{corollary}
\label{cor:covering}
Let $g\in\End(X^*)$ be an endomorphism of $X^*$ defined by a finite Mealy automaton $\A$. Let also $(\vpc_n^g)_{n\geq0}$ be the (automatic) sequence of the reduced van der Put coefficients of a transformation $\Z_d\to\Z_d$ induced by $g$. Then the underlying oriented graph $\Gamma(\B)$ of the Moore automaton $\B$ defining $(\vpc_n^g)_{n\geq0}$ obtained from $\A$ by Algorithm~\ref{alg:mealy_to_moore} covers the underlying oriented graph $\Gamma(\A)$.
\end{corollary}

\begin{proof}
Since the transitions in the original Mealy automaton $\A$ defining $g$ are defined by $\delta(g|_v,x)=g|_{vx}$, we immediately get that the map from the set of vertices of the underlying oriented graph of $\B$ to the set of vertices of the underlying oriented graph of $\A$ defined by
\[\bigl(g|_v,(\vpc^g_{\overline{vy}})_{y\in X}\bigr)\mapsto g|_v, v\in X^*\]
is a graph covering.
\end{proof}

Now we describe the procedure that constructs a Mealy automaton of an endomorphism defined by an automatic sequence generated by a given Moore automaton.

\begin{theorem}
\label{thm:covering2}
Let $g\in\End(X^*)$ be an endomorphism of $X^*$ induced by a transformation of $\Z_d$ with the sequence of reduced van der Put coefficients $(\vpc_n^g)_{n\geq0}\subset\Z_d$ generated by a finite Moore automaton $\B$ with the set of states $Q_{\B}=\{(\vpc_n^g)_{n\geq0}|_v\colon v\in X^*\}$. Then the Mealy automaton $\A=(Q_{\A},X,\delta,\lambda,q)$, where
\begin{itemize}
\item the set of states is $Q_{\A}=\left\{\bigl((\vpc^g_n)_{n\geq 0}|_v,(\vpc^{g|_v}_{i})_{i=0,1,\ldots,d-1}\bigr)\colon v\in X^*\right\}$,
\item the transition and output functions are
\begin{equation}
\label{eqn:states_labels4}
\begin{array}{rcl}
\delta\left(\bigl((\vpc^g_n)_{n\geq 0}|_v,(\vpc^{g|_v}_{i})_{i=0,1,\ldots,d-1}\bigr),x\right)&=&\bigl((\vpc^g_n)_{n\geq 0}|_{vx},(\vpc^{g|_{vx}}_{i})_{i=0,1,\ldots,d-1}\bigr),\\
\lambda\left(\bigl((\vpc^g_n)_{n\geq 0}|_v,(\vpc^{g|_v}_{i})_{i=0,1,\ldots,d-1}\bigr),x\right)&=&\vpc^{g|_v}_{\overline x}\mmod d.
\end{array}
\end{equation}
\item the initial state is $q=\bigl((\vpc^g_n)_{n\geq 0},(\vpc^{g}_{i})_{i=0,1,\ldots,d-1}\bigr)$
\end{itemize}
is finite and defines the endomorphism $g$.
\end{theorem}

\begin{proof}
The initial Mealy automaton $\A'$ defining $g$ has the set of states $Q'=\{g|_v\colon v\in X^*\}$, transition and output functions defined as
\begin{equation}
\begin{array}{rcl}
\delta'(g|_v,x)&=&g|_{vx},\\
\lambda'(g|_v,x)&=&g_v(x),
\end{array}
\end{equation}
and the initial state $g=g|_{\epsilon}$.

Since each endomorphism of $X^*$ is uniquely defined by the sequence of its reduced van der Put coefficients, we can identify $Q'$ with the set
\[\{(\vpc^{g|_v}_n)_{n\geq 0}\colon v\in X^*\}.\]
By Corollary~\ref{cor:correspondence} the sequence $(\vpc^{g|_v}_n)_{n\geq 0}$ of the reduced van der Put coefficients that defines $g|_v$ coincides starting from term $d$ with $(\vpc^g_n)_{n\geq 0}|_v$. Therefore, each state $g|_v$ of $\A'$ can be completely defined by a pair, called the \emph{label} of this state:
\begin{equation}
\label{eqn:states_labels3}
l(g|_v)=\bigl((\vpc^g_n)_{n\geq 0}|_v,(\vpc^{g|_v}_{i})_{i=0,1,\ldots,d-1}\bigr),
\end{equation}
where $(\vpc^{g|_v}_{i})_{i=0,1,\ldots,d-1}$ is the $d$-tuple of the first $d$ terms of $(\vpc_n^{g|_v})_{n\geq0}$ that corresponds to the labels of the first level of the portrait of this sequence. As in the equation~\eqref{eqn:states_labels2}, the first component of this pair defines the terms of $(\vpc_n^{g|_v})_{n\geq0}$ at level 2 and below, and the second component consists of terms of the first level.

Similarly to the case of Theorem~\ref{thm:covering}, it is possible that different labels will define the same state of $\A'$, but clearly the automaton $\A$ from the statement of the theorem also generates $g$ since its minimization coincides with $\A'$. Indeed, the set of states of $\A$ is the set of labels of states of $\A'$ and the transition and output functions in $\A$ are obtained from the corresponding functions in $\A'$ and the definition of labels.

%In particular, $\lambda'\left(\bigl((\vpc^g_n)_{n\geq 0}|_v,(\vpc^{g|_v}_{i})_{i=0,1,\ldots,d-1}\bigr),x\right)$ is defined to be the image of $x$ under the endomorphism with the sequence of reduced van der Put coefficients $\vpc^{g|_v}_{\overline x}\mmod d.

Finally, the finiteness of $Q$ follows from the above proof of Theorem~\ref{thm:characterization} since the set $\{(\vpc^g_n)_{n\geq 0}|_v\colon v\in X^*\}$ (coinciding with $Q_{\B}$) is finite, and the set $\{\vpc^{g|_v}_{i}\colon v\in X^*, i=0,1,\ldots,d-1\}$ is finite as well, which follows from Lemma~\ref{lem:finite}.
\end{proof}

We conclude with the description of the algorithm of building the Mealy automaton of an endomorphism of $X^*$ from a Moore automaton defining the sequence of its reduced van der Put coefficients.

\begin{algorithm}[Construction of Mealy automaton from Moore automaton]
\label{alg:moore_to_mealy}
Let $g\in\End(X^*)$ be an endomorphism of $X^*$ induced by a transformation of $\Z_d$ with the sequence of reduced van der Put coefficients $(\vpc_n^g)_{n\geq0}$, $\vpc_n^g\in\Z_d$ defined by a finite Moore automaton $\B$ with the set of states $Q_{\B}=\{(\vpc_n^g)_{n\geq0}|_v\colon v\in X^*\}$. To construct a Mealy automaton $\A=(Q,X,\delta,\lambda,q)$ defining endomorphism $g$ complete the following steps.
\begin{enumerate}
\item[Step 1.] Start building the set of states of $\A$ from its initial state $q=\bigl((\vpc^g_n)_{n\geq 0},(\vpc^{g}_{i})_{i=0,1,\ldots,d-1}\bigr)$ with $\tau(q)=\vpc^g_0$. Define $Q_0=\{q\}$.
\item[Step 2.] To build $Q_{i+1}$ from $Q_i$ for $i\geq 1$ start from empty set and for each state $q=\bigl((\vpc^g_n)_{n\geq 0}|_v,(\vpc^{g|_v}_{i})_{i=0,1,\ldots,d-1}\bigr)\in Q_i$ and each $x\in X$ add state $q_x=\bigl((\vpc^g_n)_{n\geq 0}|_{vx},(\vpc^{g|_{vx}}_{i})_{i=0,1,\ldots,d-1}\bigr)$ to $Q_{i+1}$ unless it belongs to $Q_{j}$ for some $j\leq i$ or is already in $Q_{i+1}$. Use the second case in~\eqref{eqn:lambda_section} to calculate $\vpc^{g|_{vx}}_{i}$ from $\vpc^{g}_{j}$'s, which are the values of the output function of the given Moore automaton. Extend the transition function by $\delta(q,x)=q_x$ and the output function by $\lambda(q,x)=\vpc^{g|_v}_{\overline{x}}\mmod d$.
\item[Step 3.] Repeat Step 2 until $Q_{i+1}=\emptyset$.
\item[Step 4.] The set of states of Mealy automaton $\A$ is $\cup_{i\geq0}Q_i$, where transition and output functions are defined in Step 2.
\end{enumerate}
\end{algorithm}

\begin{corollary}
\label{cor:covering2}
Let $g\in\End(X^*)$ be an endomorphism of $X^*$ induced by a selfmap of $\Z_d$ with the sequence of reduced van der Put coefficients defined by finite Moore automaton $\B$. Then the underlying oriented graph $\Gamma(\A)$ of the Mealy automaton $\A$ obtained from $\B$ by Algorithm~\ref{alg:moore_to_mealy}  covers the underlying oriented graph of $\B$.
\end{corollary}

\begin{proof}
Since the transitions in the original Moore automaton $\B$ defining $g$ are defined by $\delta((\vpc^g_n)_{n\geq 0}|_v,x)=(\vpc^g_n)_{n\geq 0}|_{vx}$, we immediately get that the map from the underlying oriented graph of $\A$ to the underlying oriented graph of $\B$ defined by
\[\bigl((\vpc^g_n)_{n\geq 0}|_v,(\vpc^{g|_v}_{i})_{i=0,1,\ldots,d-1}\bigr)\mapsto (\vpc^g_n)_{n\geq 0}|_v, v\in X^*\]
is a graph covering.
\end{proof}

\section{Examples}
\label{sec:examples}

\subsection{Moore automaton from Mealy automaton}
We first give the example of construction of a Moore automaton from Mealy automaton. Consider the lamplighter group $\LL=(\Z/2\Z)\wr\Z$ generated by the 2-state Mealy automaton $\A$ over the 2-letter alphabet $X=\{0,1\}$ from~\cite{grigorch_z:lamplighter} shown in Figure~\ref{fig:lamp_aut} and defined by the following wreath recursion:
\[\begin{array}{rcl}
p&=&(p,q)(01),\\
q&=&(p,q).
\end{array}
\]
%where $\sigma=(01)$ denotes the nontrivial transposition from $\Sym(X)$.

\begin{proposition}
The Moore automaton $\B_p$ generating the sequence of reduced van der Put coefficients of the transformation of $\Z_2$ induced by automorphism $p$ is shown in Figure~\ref{fig:aut_lampl_moore}, where the initial state is on top, and the value of the output function $\tau$ of $\B_p$ at a given state is equal to the first component of the pair of $d$-adic integers in its label.
\begin{figure}[h]
\begin{center}
\epsfig{file=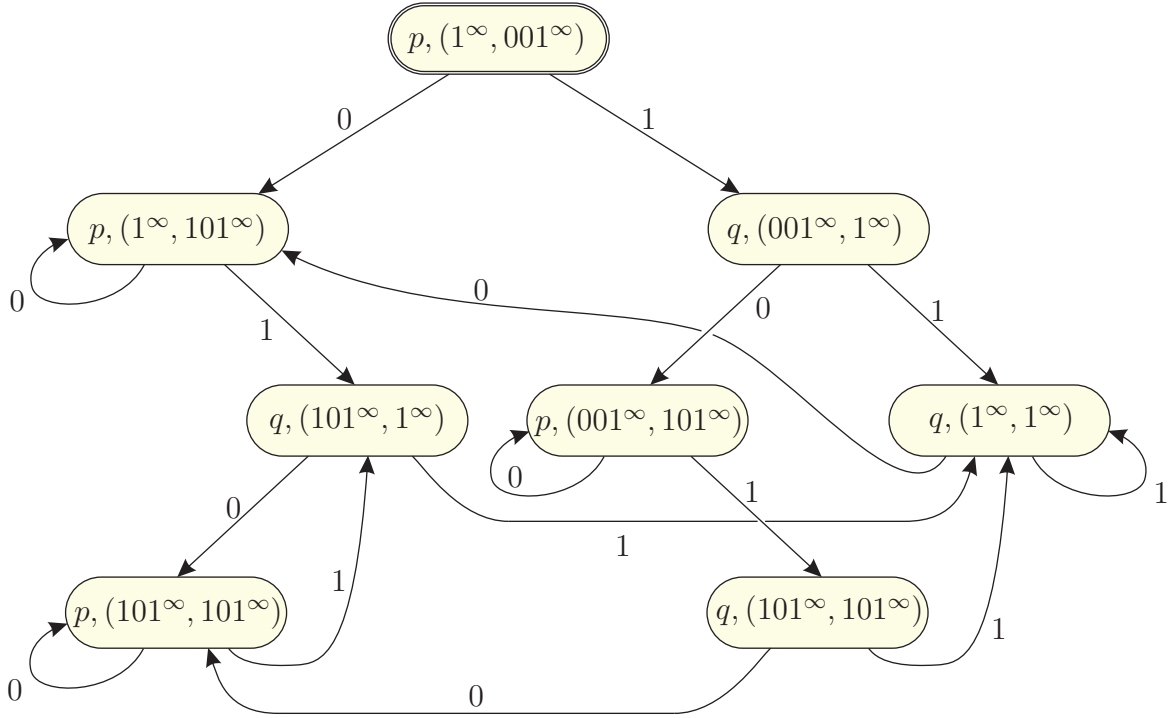}
\caption{Moore automaton $\B_p$ generating the sequence $(\vpc^p_n)_{n\geq 0}$ of reduced van der Put coefficients of the generator $p$ of the lamplighter group $\mathcal L$\label{fig:aut_lampl_moore}}
\end{center}
\end{figure}
\end{proposition}

\begin{proof}
We will apply Algorithm~\ref{alg:mealy_to_moore} and construct the sections of $(\vpc^p_n)_{n\geq0}$ at the vertices of $X^*$ in the form~\eqref{eqn:states_labels}. It may be useful to refer to Figure~\ref{fig:aut_lampl_moore} to understand better the calculations that follow.

The label of the initial state $(\vpc^p_n)_{n\geq0}|_\varepsilon$ of $\B_p$ is $\bigl(p|_\varepsilon,(\vpc^p_{0},\vpc^p_{1})\bigr)$. By equation~\eqref{eqn:lambda_n} we get:

\[\vpc^p_0=p(0^\infty)=1^\infty\qquad\text{and}\qquad\vpc^p_1=p(10^\infty)=001^\infty.\]
Therefore, the initial state of $\B_p$ is labeled by
\[l((\vpc^p_n)_{n\geq0}|_\varepsilon)=\bigl(p,(1^\infty,001^\infty)\bigr).\]
We proceed with the states corresponding to the vertices of the first level of $X^*$. We calculate:
\[\vpc^p_2=\frac{p(010^\infty)-p(0^\infty)}2=\frac{1001^\infty-1^\infty}2=101^\infty,\]
\[\vpc^p_3=\frac{p(110^\infty)-p(10^\infty)}2=\frac{0101^\infty-001^\infty}2=1^\infty.\]
Therefore, we get the labels of two more states in $\B_p$:
\[l((\vpc^p_n)_{n\geq0}|_0)=\bigl(p|_0,(\vpc^p_0,\vpc^p_2)\bigr)=\bigl(p,(1^\infty,101^\infty)\bigr),\]
\[l((\vpc^p_n)_{n\geq0}|_1)=\bigl(p|_1,(\vpc^p_1,\vpc^p_3)\bigr)=\bigl(q,(001^\infty,1^\infty)\bigr).\]
To obtain labels of the states at the vertices of deeper levels we use Corollary~\ref{cor:sections2}. Namely, for $n>3$, we have that $[n]_2=vx1\in X^*$ for some $v\in X^*$ and $x\in X$. Therefore by Corollary~\ref{cor:sections2}
\[\vpc^p_n=\vpc^p_{\overline{vx1}}=\vpc^{p|_v}_{\overline{x1}}=\vpc^{p|_v}_{n\mmod 4}.\]
Therefore, it is enough to compute the first 4 values of $(\vpc^{p|_v}_n)_{n\geq 0}$ for all states $p|_v$ of an automaton $\A$.
Since there are only 2 states in $\A$ and we have computed the first 4 values of $(\vpc^p_n)_{n\geq0}$, we proceed to $(\vpc^p_n)_{n\geq0}$:
\[\begin{array}{l}
\vpc^q_0=q(0^\infty)=01^\infty,\\
\vpc^q_1=q(10^\infty)=101^\infty,\\
\vpc^q_2=\frac{q(010^\infty)-q(0^\infty)}2=\frac{0001^\infty-01^\infty}2=101^\infty,\\
\vpc^q_3=\frac{q(110^\infty)-q(10^\infty)}2=\frac{1101^\infty-101^\infty}2=1^\infty.
\end{array}
\]
Now, by Corollary~\ref{cor:correspondence} we have that
\[\begin{array}{l}
\vpc^p_4=\vpc^p_{\overline{001}}=\vpc^{p|_0}_{\overline{01}}=\vpc^p_{\overline{01}}=\vpc^p_2=101^\infty,\\
\vpc^p_6=\vpc^p_{\overline{011}}=\vpc^{p|_0}_{\overline{11}}=\vpc^p_{\overline{11}}=\vpc^p_3=1^\infty,\\
\vpc^p_5=\vpc^p_{\overline{101}}=\vpc^{p|_1}_{\overline{01}}=\vpc^q_{\overline{01}}=\vpc^q_2=101^\infty,\\
\vpc^p_7=\vpc^p_{\overline{111}}=\vpc^{p|_1}_{\overline{11}}=\vpc^q_{\overline{11}}=\vpc^q_3=1^\infty.
\end{array}
\]

Thus, the states at the second level has the following labels:
\[\begin{array}{l}
l((\vpc^p_n)_{n\geq0}|_{00})=\bigl(p|_{00},(\vpc^p_{\overline{000}},\vpc^p_{\overline{001}})\bigr)=\bigl(p,(\vpc^p_0,\vpc^p_4)\bigr)=\bigl(p,(1^\infty,101^\infty)\bigr),\\
l((\vpc^p_n)_{n\geq0}|_{01})=\bigl(p|_{01},(\vpc^p_{\overline{010}},\vpc^p_{\overline{011}})\bigr)=\bigl(q,(\vpc^p_2,\vpc^p_6)\bigr)=\bigl(q,(101^\infty,1^\infty)\bigr),\\
l((\vpc^p_n)_{n\geq0}|_{10})=\bigl(p|_{10},(\vpc^p_{\overline{100}},\vpc^p_{\overline{101}})\bigr)=\bigl(p,(\vpc^p_1,\vpc^p_5)\bigr)=\bigl(p,(001^\infty,101^\infty)\bigr),\\
l((\vpc^p_n)_{n\geq0}|_{11})=\bigl(p|_{11},(\vpc^p_{\overline{110}},\vpc^p_{\overline{111}})\bigr)=\bigl(q,(\vpc^p_3,\vpc^p_7)\bigr)=\bigl(q,(1^\infty,1^\infty)\bigr).\\
\end{array}
\]

Since $l((\vpc^p_n)_{n\geq0}|_{00})=l((\vpc^p_n)_{n\geq0}|_{0})$, we can stop calculations along this branch. For other branches we compute similarly on the next level. We start from the branch $01$.
\[l((\vpc^p_n)_{n\geq0}|_{010})=\bigl(p|_{010},(\vpc^p_{\overline{0100}},\vpc^p_{\overline{0101}})\bigr)=\bigl(p,(\vpc^p_2,\vpc^{p|_{01}}_{\overline{01}})\bigr)=\bigl(p,(\vpc^p_2,\vpc^{q}_2)\bigr)=\bigl(p,(101^\infty,101^\infty)\bigr),\]
and
\begin{multline*}
l((\vpc^p_n)_{n\geq0}|_{011})=\bigl(p|_{011},(\vpc^p_{\overline{0110}},\vpc^p_{\overline{0111}})\bigr)=\bigl(q,(\vpc^{p|_{0}}_{\overline{11}},\vpc^{p|_{01}}_{\overline{11}})\bigr)\\
=\bigl(q,(\vpc^p_3,\vpc^{q}_3)\bigr)=\bigl(q,(1^\infty,1^\infty)\bigr)=l((\vpc^p_n)_{n\geq0}|_{11}),\\
\end{multline*}
For the branch $10$ we obtain:
\begin{multline*}
l((\vpc^p_n)_{n\geq0}|_{100})=\bigl(p|_{100},(\vpc^p_{\overline{1000}},\vpc^p_{\overline{1001}})\bigr)=\bigl(p,(\vpc^p_1,\vpc^{p|_{10}}_{\overline{01}})\bigr)\\
=\bigl(p,(\vpc^p_1,\vpc^{p}_2)\bigr)=\bigl(p,(001^\infty,101^\infty)\bigr)=l((\vpc^p_n)_{n\geq0}|_{10})
\end{multline*}
and
\[
l((\vpc^p_n)_{n\geq0}|_{101})=\bigl(p|_{101},(\vpc^p_{\overline{1010}},\vpc^p_{\overline{1011}})\bigr)=\bigl(q,(\vpc^{p|_{1}}_{\overline{01}},\vpc^{p|_{10}}_{\overline{11}})\bigr)=\bigl(q,(\vpc^q_2,\vpc^{p}_2)\bigr)=\bigl(q,(101^\infty,101^\infty)\bigr).
\]

For the branch $11$ we get:
\begin{multline*}
l((\vpc^p_n)_{n\geq0}|_{110})=\bigl(p|_{110},(\vpc^p_{\overline{1100}},\vpc^p_{\overline{1101}})\bigr)=\bigl(p,(\vpc^p_3,\vpc^{p|_{11}}_{\overline{01}})\bigr)\\
=\bigl(p,(\vpc^p_3,\vpc^{q}_2)\bigr)=\bigl(p,(1^\infty,101^\infty)\bigr)=l((\vpc^p_n)_{n\geq0}|_{0})
\end{multline*}
and
\begin{multline*}
l((\vpc^p_n)_{n\geq0}|_{111})=\bigl(p|_{111},(\vpc^p_{\overline{1110}},\vpc^p_{\overline{1111}})\bigr)=\bigl(q,(\vpc^{p|_{1}}_{\overline{11}},\vpc^{p|_{11}}_{\overline{11}})\bigr)\\
=\bigl(q,(\vpc^q_3,\vpc^{q}_3)\bigr)=\bigl(q,(1^\infty,1^\infty)\bigr)=l((\vpc^p_n)_{n\geq0}|_{11}).
\end{multline*}

At this moment we have two unfinished branches: $010$ and $101$. For $010$ we have:
\begin{multline*}
l((\vpc^p_n)_{n\geq0}|_{0100})=\bigl(p|_{0100},(\vpc^p_{\overline{01000}},\vpc^p_{\overline{01001}})\bigr)=\bigl(p,(\vpc^p_2,\vpc^{p|_{010}}_{\overline{01}})\bigr)\\
=\bigl(p,(\vpc^p_2,\vpc^{p}_2)\bigr)=\bigl(p,(101^\infty,101^\infty)\bigr)=l((\vpc^p_n)_{n\geq0}|_{010})
\end{multline*}
and
\begin{multline*}
l((\vpc^p_n)_{n\geq0}|_{0101})=\bigl(p|_{0101},(\vpc^p_{\overline{01010}},\vpc^p_{\overline{01011}})\bigr)=\bigl(q,(\vpc^{p|_{01}}_{\overline{01}},\vpc^{p|_{010}}_{\overline{11}})\bigr)\\
=\bigl(q,(\vpc^q_2,\vpc^{p}_3)\bigr)=\bigl(q,(101^\infty,1^\infty)\bigr)=l((\vpc^p_n)_{n\geq0}|_{01}).
\end{multline*}

Finally, for $101$ branch we compute:
\begin{multline*}
l((\vpc^p_n)_{n\geq0}|_{1010})=\bigl(p|_{1010},(\vpc^p_{\overline{10100}},\vpc^p_{\overline{10101}})\bigr)=\bigl(p,(\vpc^{p|_{1}}_{\overline{01}},\vpc^{p|_{101}}_{\overline{01}})\bigr)\\
=\bigl(p,(\vpc^q_2,\vpc^{q}_2)\bigr)=\bigl(p,(101^\infty,101^\infty)\bigr)=l((\vpc^p_n)_{n\geq0}|_{010})
\end{multline*}
and
\begin{multline*}
l((\vpc^p_n)_{n\geq0}|_{1011})=\bigl(p|_{1011},(\vpc^p_{\overline{10110}},\vpc^p_{\overline{10111}})\bigr)=\bigl(q,(\vpc^{p|_{10}}_{\overline{11}},\vpc^{p|_{101}}_{\overline{11}})\bigr)\\
=\bigl(q,(\vpc^p_3,\vpc^{q}_3)\bigr)=\bigl(q,(1^\infty,1^\infty)\bigr)=l((\vpc^p_n)_{n\geq0}|_{11}).
\end{multline*}
We have completed all the branches and constructed all the transitions in the automaton $\B_p$.
\end{proof}

\subsection{Mealy automaton from Moore automaton}
In this subsection we provide an example of the converse construction. Namely, we will construct the finite state endomorphism of $\{0,1\}^*$ that induces a transformation of $\Z_2$ with the Thue-Morse sequence of reduced van der Put coefficients, where we treat $0$ as $0^\infty$ and $1$ as $10^\infty$ according to the standard embedding of $\Z$ into $\Z_2$.

Recall, that the Thue-Morse sequence $(t_n)_{n\geq0}$ is the binary sequence defined by a Moore automaton shown in Figure~\ref{fig:thue_morse_moore}. It can be obtained by starting with 0 and successively appending the Boolean complement of the sequence obtained thus far. The first 32 values of this sequence are shown in Table~\ref{tab:thue}.

\begin{table}
\[
\begin{array}{|l|l|l|||l|l|l|||l|l|l|||l|l|l|}
\hline
n&[n]_2&t_n&n&[n]_2&t_n&n&[n]_2&t_n&n&[n]_2&t_n\\ \hline
0&0  &0  &8 &0001&1  &16&00001&1  &24&00011&0\\
1&1  &1  &9 &1001&0  &17&10001&0  &25&10011&1\\
2&01 &1  &10&0101&0  &18&01001&0  &26&01011&1\\
3&11 &0  &11&1101&1  &19&11001&1  &27&11011&0\\
4&001&1  &12&0011&0  &20&00101&0  &28&00111&1\\
5&101&0  &13&1011&1  &21&10101&1  &29&10111&0\\
6&011&0  &14&0111&1  &22&01101&1  &30&01111&0\\
7&111&1  &15&1111&0  &23&11101&0  &31&11111&1\\ \hline
\end{array}
\]
\caption{The first 32 values of the Thue-Morse sequence\label{tab:thue}}
\end{table}

\begin{figure}[h]
\begin{center}
\epsfig{file=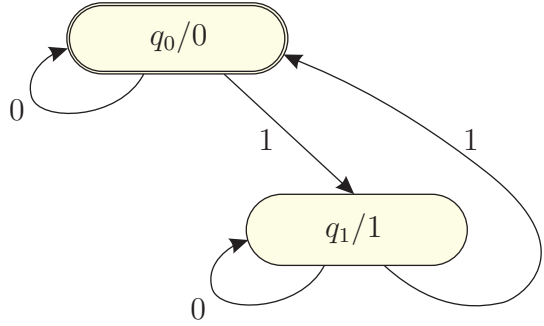}
\caption{Moore automaton $\B$ generating the Thue-Morse sequence\label{fig:thue_morse_moore}}
\end{center}
\end{figure}

\begin{proposition}
\label{prop:thue}
Endomorphism $t$ of $X^*$ inducing a transformation of $\Z_2$ with the Thue-Morse sequence $(\vpc^t_n)_{n\geq 0}=(t_n)_{n\geq 0}$ of the reduced van der Put coefficients is defined by the 2-state Mealy automaton $\A_t$ shown in Figure~\ref{fig:thue_morse_mealy} with the following wreath recursion:
\[\begin{array}{lll}
t&=&(t,s),\\
s&=&(s,t) {{0 1}\choose{0 0}},
\end{array}\]
where ${{0 1}\choose{0 0}}$ denotes the selfmap of $\{0,1\}$ sending both of its elements to $0$.
\end{proposition}

\begin{figure}[h]
\begin{center}
\epsfig{file=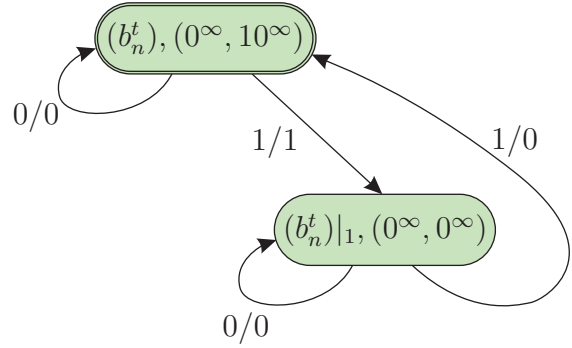}
\caption{Mealy automaton $\A_t$ defining a transformation of $\Z_2$ whose sequence of reduced van der Put coefficients is the Thue-Morse sequence\label{fig:thue_morse_mealy}}
\end{center}
\end{figure}

\begin{proof}
We will follow Algorithm~\ref{alg:moore_to_mealy}, according to which the states of $\A_t$ are the pairs of the form
\begin{equation}
\label{eqn:states_labels_t}
l(t|_v)=\bigl((\vpc^t_n)_{n\geq 0}|_v,(\vpc^{t|_v}_{0}, \vpc^{t|_v}_{1})\bigr),
\end{equation}

Below we will suppress the subscript $n\geq 0$ in the notation for sequences to simplify the exposition. E.g., we will write simply $(\vpc^t_n)$ for $(\vpc^t_n)_{n\geq 0}$.

The initial state $t$ will have a label
\[l(t)=l(t|_\varepsilon)=\bigl((\vpc^t_n),(\vpc^{t}_{0}, \vpc^{t}_{1})\bigr)=\bigl((\vpc^t_n),(0^\infty, 10^\infty)\bigr).\]

We proceed to calculating the labels of the sections at the vertices of the first level. Using Theorem~\ref{thm:van_der_put_sections} (namely, the first two cases in equation~\eqref{eqn:lambda_section}), and the values of $\vpc^t_n=t_n$ of the Thue-Morse sequence from Table~\ref{tab:thue}, we obtain:

\begin{multline*}
l(t|_0)=\bigl((\vpc^t_n)|_0,(\vpc^{t|_0}_{0}, \vpc^{t|_0}_{1})\bigr)=\bigl((\vpc^t_n),(\sigma(\vpc^t_0), \vpc^t_2+\sigma(\vpc^t_0))\bigr)\\
=\bigl((\vpc^t_n),(\sigma(0^\infty), 10^\infty+\sigma(0^\infty))\bigr)=\bigl((\vpc^t_n),(0^\infty, 10^\infty)\bigr)=l(t).
\end{multline*}
We also used above the fact that $(\vpc^t_n)|_0=(\vpc^t_n)$ that follows from the structure of automaton $\B$. Therefore, we can stop developing the branch that starts with 0 and move to the branch starting from 1. Similarly we get

\begin{multline}
\label{eqn:vpc1}
l(t|_1)=\bigl((\vpc^t_n)|_1,(\vpc^{t|_1}_{0}, \vpc^{t|_1}_{1})\bigr)=\bigl((\vpc^t_n)|_1,(\sigma(\vpc^t_1), \vpc^t_3+\sigma(\vpc^t_1))\bigr)\\
=\bigl((\vpc^t_n)|_1,(\sigma(10^\infty), 0^\infty+\sigma(10^\infty))\bigr)=\bigl((\vpc^t_n)|_1,(0^\infty, 0^\infty)\bigr),
\end{multline}
so we obtained a new section. We compute the sections at the vertices of the second level using Figure~\ref{fig:portrait_correspondence}, keeping in mind that according to equation~\eqref{eqn:vpc1} $\vpc^{t|_1}_0=0^\infty$:

\begin{multline*}
l(t|_{10})=\bigl((\vpc^t_n)|_{10},(\vpc^{t|_{10}}_{0}, \vpc^{t|_{10}}_{1})\bigr)=\bigl((\vpc^t_n)|_{1},(\vpc^{(t|_1)|_{0}}_{0}, \vpc^{(t|_1)|_{0}}_{1})\bigr)\\
=\bigl((\vpc^t_n)|_1,(\sigma(\vpc^{t|_1}_0), \vpc^{t|_1}_2+\sigma(\vpc^{t|_1}_0))\bigr)=\bigl((\vpc^t_n)|_1,(\sigma(0^\infty), \vpc^{t}_5+\sigma(0^\infty))\bigr)\\
=\bigl((\vpc^t_n)|_1,(0^\infty, 0^\infty+0^\infty)\bigr)=\bigl((\vpc^t_n)|_1,(0^\infty, 0^\infty)\bigr)=l(t|_1).
\end{multline*}

Finally, since according to equation~\eqref{eqn:vpc1} $\vpc^{t|_1}_1=0^\infty$, we calculate the last section at $11$:
\begin{multline*}
l(t|_{11})=\bigl((\vpc^t_n)|_{11},(\vpc^{t|_{11}}_{0}, \vpc^{t|_{11}}_{1})\bigr)=\bigl((\vpc^t_n),(\vpc^{(t|_1)|_{1}}_{0}, \vpc^{(t|_1)|_{1}}_{1})\bigr)\\
=\bigl((\vpc^t_n),(\sigma(\vpc^{t|_1}_1), \vpc^{t|_1}_3+\sigma(\vpc^{t|_1}_1))\bigr)=\bigl((\vpc^t_n),(\sigma(0^\infty), \vpc^{t}_7+\sigma(0^\infty))\bigr)\\
=\bigl((\vpc^t_n),(0^\infty, 10^\infty+0^\infty)\bigr)=\bigl((\vpc^t_n),(0^\infty, 10^\infty)\bigr)=l(t).
\end{multline*}

We have completed all the branches and constructed all the transitions in the automaton $\A_t$. We only need now to compute the values of the output function. By equation~\eqref{eqn:states_labels4} we get:
\[
\begin{array}{l}
\lambda\left(\bigl((\vpc^t_n),(0^\infty,10^\infty)\bigr),0\right)=0^\infty\mmod 2=0,\\
\lambda\left(\bigl((\vpc^t_n),(0^\infty,10^\infty)\bigr),1\right)=10^\infty\mmod 2=1,\\
\lambda\left(\bigl((\vpc^t_n)|_1,(0^\infty,0^\infty)\bigr),0\right)=0^\infty\mmod 2=0,\\
\lambda\left(\bigl((\vpc^t_n)|_1,(0^\infty,0^\infty)\bigr),1\right)=0^\infty\mmod 2=0,
\end{array}
\]
which completes the proof of the proposition.
\end{proof}

\noindent{\textbf{Acknowledgements.} The authors are thankful to Zoran \v Sun\'ic and Svetlana Katok for enlightening  conversations on the subject of the paper. The first author graciously acknowledges support from the Simons Foundation through Collaboration Grant \#527814 and also is supported by the mega-grant of the Russian Federation Government (N14.W03.31.0030). He also gratefully acknowledges support of the Swiss National Science Foundation. The second author greatly appreciates the support of the Simons Foundation through Collaboration Grant \#317198. The work on this project was partially conducted during the authors' visits to American Institute of Mathematics SQuaRE program; they thank the institute for hospitality and support.}

\bibliographystyle{plain}

\begin{thebibliography}{10}

\bibitem{ahmed_s:polynomial_ergodicity}
Elsayed Ahmed and Dmytro Savchuk.
\newblock Endomorphisms of regular rooted trees induced by the action of
  polynomials on the ring $\mathbb {Z}_d$ of $d$-adic integers.
\newblock {\em J. Algebra Appl.}, Published online, 2019.

\bibitem{allouche:automatic_sequences03}
Jean-Paul Allouche and Jeffrey Shallit.
\newblock {\em Automatic sequences}.
\newblock Cambridge University Press, Cambridge, 2003.
\newblock Theory, applications, generalizations.

\bibitem{anashin:automata12}
V.~Anashin.
\newblock Automata finiteness criterion in terms of van der {P}ut series of
  automata functions.
\newblock {\em p-Adic Numbers Ultrametric Anal. Appl.}, 4(2):151--160, 2012.

\bibitem{anashin_ky:characterization_of_ergodic11}
V.~S. Anashin, A.~Yu. Khrennikov, and E.~I. Yurova.
\newblock Characterization of ergodic {$p$}-adic dynamical systems in terms of
  the van der {P}ut basis.
\newblock {\em Dokl. Akad. Nauk}, 438(2):151--153, 2011.
\newblock [{E}nglish translation: \textit{{D}okl. {M}ath.} 83 (2011), no. 3,
  306-308].

\bibitem{ana:erg}
Vladimir Anashin.
\newblock Ergodic transformations in the space of {$p$}-adic integers.
\newblock In {\em {$p$}-adic mathematical physics}, volume 826 of {\em AIP
  Conf. Proc.}, pages 3--24. Amer. Inst. Phys., Melville, NY, 2006.

\bibitem{bartholdi_g:spectrum}
L.~Bartholdi and R.~I. Grigorchuk.
\newblock On the spectrum of {H}ecke type operators related to some fractal
  groups.
\newblock {\em Tr. Mat. Inst. Steklova}, 231(Din. Sist., Avtom. i Beskon.
  Gruppy):5--45, 2000.

\bibitem{bartholdi_gn:fractal}
Laurent Bartholdi, Rostislav Grigorchuk, and Volodymyr Nekrashevych.
\newblock From fractal groups to fractal sets.
\newblock In {\em Fractals in Graz 2001}, Trends Math., pages 25--118.
  Birkh\"auser, Basel, 2003.

\bibitem{bartholdi_n:rabbit}
Laurent~I. Bartholdi and Volodymyr~V. Nekrashevych.
\newblock Thurston equivalence of topological polynomials.
\newblock {\em Acta Math.}, 197(1):1--51, 2006.

\bibitem{bartholdi_s:bsolitar}
Laurent~I. Bartholdi and Zoran {\v{S}}uni{\'k}.
\newblock Some solvable automaton groups.
\newblock In {\em Topological and Asymptotic Aspects of Group Theory}, volume
  394 of {\em Contemp. Math.}, pages 11--29. Amer. Math. Soc., Providence, RI,
  2006.

\bibitem{bernstein_l:3n_plus_1_96}
Daniel~J. Bernstein and Jeffrey~C. Lagarias.
\newblock The {$3x+1$} conjugacy map.
\newblock {\em Canad. J. Math.}, 48(6):1154--1169, 1996.

\bibitem{cain:automaton_semigroups}
Alan~J. Cain.
\newblock Automaton semigroups.
\newblock {\em Theoret. Comput. Sci.}, 410(47-49):5022--5038, 2009.

\bibitem{cull_n:hanoi_codes99}
Paul Cull and Ingrid Nelson.
\newblock Error-correcting codes on the {T}owers of {H}anoi graphs.
\newblock {\em Discrete Math.}, 208/209:157--175, 1999.
\newblock Combinatorics (Assisi, 1996).

\bibitem{garzon_z:crypto}
Max Garzon and Yechezkel Zalcstein.
\newblock The complexity of {G}rigorchuk groups with application to
  cryptography.
\newblock {\em Theoret. Comput. Sci.}, 88(1):83--98, 1991.

\bibitem{goresky_k:alg_shift_register_sequences12}
Mark Goresky and Andrew Klapper.
\newblock {\em Algebraic shift register sequences}.
\newblock Cambridge University Press, Cambridge, 2012.

\bibitem{grigorch_lns:self-similar_groups_automatic_sequences16}
R.~Grigorchuk, Y.~Leonov, V.~Nekrashevych, and V.~Sushchansky.
\newblock Self-similar groups, automatic sequences, and unitriangular
  representations.
\newblock {\em Bull. Math. Sci.}, 6(2):231--285, 2016.

\bibitem{grigorch:burnside}
R.~I. Grigorchuk.
\newblock On {B}urnside's problem on periodic groups.
\newblock {\em Funktsional. Anal. i Prilozhen.}, 14(1):53--54, 1980.

\bibitem{grigorch:milnor}
R.~I. Grigorchuk.
\newblock On the {M}ilnor problem of group growth.
\newblock {\em Dokl. Akad. Nauk SSSR}, 271(1):30--33, 1983.

\bibitem{gns00:automata}
R.~I. Grigorchuk, V.~V. Nekrashevich, and V.~I. Sushchanski{\u\i}.
\newblock Automata, dynamical systems, and groups.
\newblock {\em Tr. Mat. Inst. Steklova}, 231(Din. Sist., Avtom. i Beskon.
  Gruppy):134--214, 2000.

\bibitem{grigorch_ln:spectra_schreier_schroedinger18}
Rostislav Grigorchuk, Daniel Lenz, and Tatiana Nagnibeda.
\newblock Spectra of {S}chreier graphs of {G}rigorchuk's group and
  {S}chroedinger operators with aperiodic order.
\newblock {\em Math. Ann.}, 370(3-4):1607--1637, 2018.

\bibitem{grigorchuk-s:standrews}
Rostislav Grigorchuk and Zoran {\v S}uni{\'c}.
\newblock Self-similarity and branching in group theory.
\newblock In {\em Groups St. Andrews 2005, I}, volume 339 of {\em London Math.
  Soc. Lecture Note Ser.}, pages 36--95. Cambridge Univ. Press, Cambridge,
  2007.

\bibitem{grigorch_s:hanoi_spectrum}
Rostislav Grigorchuk and Zoran {\v{S}}uni{\'c}.
\newblock Schreier spectrum of the {H}anoi {T}owers group on three pegs.
\newblock In {\em Analysis on graphs and its applications}, volume~77 of {\em
  Proc. Sympos. Pure Math.}, pages 183--198. Amer. Math. Soc., Providence, RI,
  2008.

\bibitem{grigorchuk-s:hanoi-cr}
Rostislav Grigorchuk and Zoran {\v{S}}uni{\'k}.
\newblock Asymptotic aspects of {S}chreier graphs and {H}anoi {T}owers groups.
\newblock {\em C. R. Math. Acad. Sci. Paris}, 342(8):545--550, 2006.

\bibitem{grigorch_lsz:atiyah}
Rostislav~I. Grigorchuk, Peter Linnell, Thomas Schick, and Andrzej {\.Z}uk.
\newblock On a question of {A}tiyah.
\newblock {\em C. R. Acad. Sci. Paris S\'er. I Math.}, 331(9):663--668, 2000.

\bibitem{grigorch_z:lamplighter}
Rostislav~I. Grigorchuk and Andrzej {\.Z}uk.
\newblock The lamplighter group as a group generated by a 2-state automaton,
  and its spectrum.
\newblock {\em Geom. Dedicata}, 87(1-3):209--244, 2001.

\bibitem{grigorch_z:basilica}
Rostislav~I. Grigorchuk and Andrzej {\.Z}uk.
\newblock On a torsion-free weakly branch group defined by a three state
  automaton.
\newblock {\em Internat. J. Algebra Comput.}, 12(1-2):223--246, 2002.

\bibitem{katok:p-adic_analysis07}
Svetlana Katok.
\newblock {\em {$p$}-adic analysis compared with real}, volume~37 of {\em
  Student Mathematical Library}.
\newblock American Mathematical Society, Providence, RI; Mathematics Advanced
  Study Semesters, University Park, PA, 2007.

\bibitem{larin:tr}
M.~V. Larin.
\newblock Transitive polynomial transformations of residue rings.
\newblock {\em Diskret. Mat.}, 14(2):20--32, 2002.

\bibitem{mahler:p-adic}
Kurt Mahler.
\newblock {\em {$p$}-adic numbers and their functions}, volume~76 of {\em
  Cambridge Tracts in Mathematics}.
\newblock Cambridge University Press, Cambridge-New York, second edition, 1981.

\bibitem{miasnikov_s:automatic_graph}
Alexei Miasnikov and Dmytro Savchuk.
\newblock An example of an automatic graph of intermediate growth.
\newblock {\em Ann. Pure Appl. Logic}, 166(10):1037--1048, 2015.

\bibitem{miasnikov_s:cayley_automatic11}
Alexei Miasnikov and Zoran {\v{S}}uni{\'c}.
\newblock Cayley graph automatic groups are not necessarily {C}ayley graph
  biautomatic.
\newblock In {\em Language and automata theory and applications}, volume 7183
  of {\em Lecture Notes in Comput. Sci.}, pages 401--407. Springer, Heidelberg,
  2012.

\bibitem{myasnikov_su:non_commutative_crypto_book11}
Alexei Myasnikov, Vladimir Shpilrain, and Alexander Ushakov.
\newblock {\em Non-commutative cryptography and complexity of group-theoretic
  problems}, volume 177 of {\em Mathematical Surveys and Monographs}.
\newblock American Mathematical Society, Providence, RI, 2011.
\newblock With an appendix by Natalia Mosina.

\bibitem{myasnikov_u:random_subgroups08}
Alexei~G. Myasnikov and Alexander Ushakov.
\newblock Random subgroups and analysis of the length-based and quotient
  attacks.
\newblock {\em J. Math. Cryptol.}, 2(1):29--61, 2008.

\bibitem{nekrash_s:12endomorph}
V.~Nekrashevych and S.~Sidki.
\newblock Automorphisms of the binary tree: state-closed subgroups and dynamics
  of $1/2$-endomorphisms.
\newblock volume 311 of {\em London Math. Soc. Lect. Note Ser.}, pages
  375--404. {Cambridge Univ. Press}, 2004.

\bibitem{nekrash:self-similar}
Volodymyr Nekrashevych.
\newblock {\em Self-similar groups}, volume 117 of {\em Mathematical Surveys
  and Monographs}.
\newblock American Mathematical Society, Providence, RI, 2005.

\bibitem{petrides:cryptoanalysis_grigorchuk}
George Petrides.
\newblock Cryptanalysis of the public key cryptosystem based on the word
  problem on the {G}rigorchuk groups.
\newblock In {\em Cryptography and coding}, volume 2898 of {\em Lecture Notes
  in Comput. Sci.}, pages 234--244. Springer, Berlin, 2003.

\bibitem{schikhof:ultrametric84}
W.~H. Schikhof.
\newblock {\em Ultrametric calculus}, volume~4 of {\em Cambridge Studies in
  Advanced Mathematics}.
\newblock Cambridge University Press, Cambridge, 1984.
\newblock An introduction to $p$-adic analysis.

\bibitem{sholomov:osnovy_diskr80}
L.A. Sholomov.
\newblock {\em Foundations of the theory of discrete logical and computational
  devices (In Russian)}.
\newblock Moscow: Nauka, 1980.

\bibitem{walsh:basis23}
J.~L. Walsh.
\newblock A {C}losed {S}et of {N}ormal {O}rthogonal {F}unctions.
\newblock {\em Amer. J. Math.}, 45(1):5--24, 1923.

\end{thebibliography}
\def\cprime{$'$} \def\cydot{\leavevmode\raise.4ex\hbox{.}} \def\cprime{$'$}
  \def\cprime{$'$} \def\cprime{$'$} \def\cprime{$'$} \def\cprime{$'$}
  \def\cprime{$'$} \def\cprime{$'$} \def\cprime{$'$} \def\cprime{$'$}
  \def\cprime{$'$} \def\cprime{$'$} \def\cprime{$'$} \def\cprime{$'$}
  \def\cprime{$'$}

\end{document}